\documentclass[11pt]{article}
\usepackage{complexity}

\usepackage[english]{babel}
\usepackage[utf8x]{inputenc}
\usepackage[T1]{fontenc}

 \usepackage[margin=1.in,marginparwidth=1.75cm]{geometry}

\usepackage[algo2e, ruled, vlined]{algorithm2e}
\usepackage{amsmath}
\usepackage{amsfonts}
\usepackage{graphicx}
\usepackage{amsthm}
\usepackage{dsfont}
\usepackage{bbm}

\usepackage{amssymb}

\usepackage{enumitem}

\usepackage[colorlinks=true, allcolors=blue]{hyperref}

\usepackage{xcolor} 

\SetCommentSty{mycommfont}

\theoremstyle{theorem}
\newtheorem{theorem}{Theorem}[section]

\newtheorem{definition}[theorem]{Definition}

\newtheorem{lemma}[theorem]{Lemma}

\newtheorem{corollary}[theorem]{Corollary}

\newtheorem{remark}[theorem]{Remark}

\newcommand{\Ber}{\mathbf{Ber}}
\newcommand{\Lap}{\mathbf{Lap}}
\newcommand{\calA}{\mathcal{A}}
\newcommand{\AAA}{\mathcal{A}}

\newcommand{\calM}{\mathcal{M}}

\newcommand{\calX}{{\mathcal{X}}}
\newcommand{\calY}{{\mathcal{Y}}}
\newcommand{\calW}{{\mathcal{W}}}

\newcommand{\eps}{\varepsilon}

\newcommand{\SSS}{\mathcal{S}}

\newcommand{\supp}{\mathsf{supp}}

\newcommand{\xin}[1]{{\color{red} \tiny #1 --Xin}}
\renewcommand{\xin}[1]{}

\newcommand{\edith}[1]{{\color{purple} Edith: #1}}
\renewcommand{\edith}[1]{}

\newcommand{\ignore}[1]{}

\date{}

\title{The Target-Charging Technique \\ for Privacy Analysis across Interactive Computations}
\author{{\normalfont Edith Cohen}\thanks{Google Research and Tel Aviv University. \texttt{edith@cohenwang.com}.} \and Xin Lyu\thanks{UC Berkeley and Google Research. \texttt{lyuxin1999@gmail.com}.}}

\begin{document}






 \maketitle

\begin{abstract}
We propose the \emph{Target Charging Technique} (TCT), a unified privacy analysis framework for interactive settings where a sensitive dataset is accessed multiple times using differentially private algorithms. Unlike traditional composition, where privacy guarantees deteriorate quickly with the number of accesses, TCT allows computations that don't hit a specified \emph{target}, often the vast majority, to be essentially free (while incurring instead a small overhead on those that do hit their targets). TCT generalizes tools such as the sparse vector technique and top-$k$ selection from private candidates and extends their remarkable privacy enhancement benefits from noisy Lipschitz functions to general private algorithms.

\end{abstract}

\SetKwFunction{AboveThreshold}{AboveThreshold}
\SetKwFunction{NotPrior}{NotPrior}
\SetKwFunction{Boundary}{Boundary}
\SetKwFunction{BetweenThresholds}{BetweenThresholds}

\section{Introduction}

  \SetKwFunction{CR}{ConditionalRelease}
  \SetKwFunction{RCR}{ReviseCR}

In many practical settings of data analysis and optimization, the dataset $D$ is accessed multiple times interactively via different algorithms $(\calA_i)$, so that $\calA_i$ depends on the transcript of prior responses $(\calA_j(D))_{j<i}$.   When each $\calA_i$ is privacy-preserving, we are interested in tight end-to-end privacy analysis.
We consider the standard statistical
framework of differential privacy introduced in~\cite{DMNS06}.
{\em Composition} theorems~\cite{DBLP:conf/focs/DworkRV10} are a generic way to do that and achieve overall privacy cost that scales linearly or (via  ``advanced" composition) with square-root dependence in the number of private computations. We aim for a broad understanding of scenarios where the overall privacy bounds can be lowered significantly via the following paradigm:  Each computation is specified by a private algorithm $\calA_i$ together with a {\em target} $\top_i$, that is a subset of its potential outputs.  The total privacy cost depends only on computations where 
the output hits its target, that is $\calA_i(D)\in \top_i$.
This paradigm is suitable and can be highly beneficial when (i)~the specified targets are a good proxy for the actual privacy exposure and (ii)~we expect the majority of computations to not hit their target, and thus essentially be ``free'' in terms of privacy cost.

The Sparse Vector Technique (SVT) \cite{DNRRV:STOC2009,DBLP:conf/stoc/RothR10,DBLP:conf/focs/HardtR10,DBLP:books/sp/17/Vadhan17-dp-complex} is the quintessential special case.  
SVT is focused on specific type of computations that have the form of approximate threshold tests applied to Lipschitz functions.  Concretely, each such \AboveThreshold test is specified by a
$1$-Lipschitz function $f$ and a threshold value $t$ and we wish to test 
whether $f(D) \gtrsim t$. 
The textbook SVT algorithm compares a 
noisy value with a noisy threshold (independent Laplace noise for the values and threshold noise that can be updated only after positive responses). Remarkably, the overall privacy cost depends only on the number of positive responses, roughly, composition is applied to twice the number of positive responses instead of to the total number of computations. Using our terminology, the target of each test is a positive response. 

SVT privacy analysis benefits when the majority of \AboveThreshold test results are negative (and hence ``free''). This makes SVT a key ingredient in a range of methods~\cite{DworkRothBook2014}: private multiplicative weights~\cite{DBLP:conf/focs/HardtR10}, Propose-Test-Release~\cite{DworkLei:STOC2009}, fine privacy analysis via 
distance-to-stability~\cite{pmlr-v30-Guha13},
model-agnostic private learning~\cite{BassilyTT:NEURIPS2018},
and designing streaming algorithms that are robust to adaptive inputs~\cite{HassidimKMMS20,CLNSSS:ICML2022}.\footnote{Robustness was linked to privacy so that use of SVT allowed dependence on \emph{changes to the output} rather than on the typically much larger number of updates to the input.}


We aim to extend such SVT-like privacy analysis benefits to interactive applications of \emph{general} private algorithms (that is, algorithms that provide privacy guarantees but have no other assumptions): private tests, where we would hope to incur privacy cost only for positive responses, and private algorithms that return more complex outputs, e.g., vector average, cluster centers, a sanitized dataset, 
or a trained ML model, where the goal is to incur privacy cost only when the output satisfies some criteria. 
The textbook SVT, however, seems less amenable to such extensions:
First, SVT departs from the natural paradigm of applying private algorithms to the dataset and reporting the output. A natural implementation of private \AboveThreshold tests would add Laplace noise to the value and compare with the threshold. Instead, 
SVT takes as input the Lipschitz output of the non-private algorithms with threshold value and the privacy treatment is integrated (added noise both to values and threshold). The overall utility and privacy of the complete interaction are analyzed with respect to the non-private values, which is not suitable when the algorithms are already private. 
Furthermore, the technique of using a hidden shared threshold noise across multiple \AboveThreshold tests\footnote{We mention that\cite{DBLP:conf/focs/HardtR10} did not use noisy thresholds but nearly all followup works did} is specific for Lipschitz functions, introduces dependencies between responses (that are biased the same way and can be undesirable for downstream applications), and more critically, implies additional privacy cost for reporting noisy values.
Analytics tasks often require a value to be reported with an above-threshold test result, which incurs an additional separate privacy charge with SVT~\cite{DBLP:journals/pvldb/LyuSL17}.\footnote{Reporting the noisy value that was compared with the noisy threshold discloses information on the shared threshold noise.}

Private tests, mentioned above, are perhaps the most basic extension for which we seek SVT-like benefits.  The natural approach would be to apply each test once, report the result, and hope to incur privacy charge only on positive responses. Private testing was considered in prior works~\cite{LiuT19-private-select,CLNSS:ITCS2023} but in ways that significantly departed from this natural paradigm:
Instead, the approach of \cite{LiuT19-private-select} processed the private tests so that a positive answer is returned only when the probability $p$ of a  positive response by the private test is very close to 1.\footnote{
Assuming a {\em percentile oracle}, that provides the probability $p$ of a $1$ response, they applied \AboveThreshold to 
$\log(p/(1-p)$ (that is $2\eps$-Lipschitz when the test is $\eps$-private). Note that by adding $\Lap(1/\eps)$ noise to the log ratio, we effectively need to use a threshold that applies only for $p$ that is extremely close to $1$, that is, $1-p \approx 2^{-1/\eps}$. Therefore positive responses are reported with a different (and much lower) probability than the original test. For the case where a percentile oracle is not available, \cite{LiuT19-private-select} proposed an approximate DP computation with a somewhat less efficient bound.} This seems unsatisfactory: If the design goal of the private testing algorithm was to report only very high probabilities, then this could have been integrated into the design (possibly while avoiding the factor-2 privacy overhead), and if otherwise, then we miss out on acceptable positive responses with moderately high probabilities (e.g. 95\%). 

\emph{Top-$k$ selection} is another setting
where careful privacy analysis is hugely beneficial. Top-$k$ is  a basic subroutine in data analysis, where input algorithms $(\calA_i)_{i\in[m]}$ (aka candidates) that return results with quality scores are provided in a \emph{batch} (i.e., non interactively). The selection returns the $k$ candidates with highest quality scores on our dataset. The respective private construct, where the data is sensitive and the algorithms are private, had been intensely studied~\cite{McSherryMironov:KDD2009,FriendmanSchuster:KDD2010,SteinkeUllman:FOCS2017}. We might hope for privacy cost that is close to a composition over $k$ private computations, instead of over $m \gg k$.
The natural approach for top-$k$ (and what we would do on non-sensitive data) is \emph{one-shot} (Algorithm~\ref{algo:top-k}), where each algorithm is applied once and the responses with top-$k$ scores are reported. 
Prior works on private selection that achieve this analysis goal include those 
\cite{DBLP:conf/nips/DurfeeR19,DBLP:conf/icml/QiaoSZ21}
that use the natural one-shot selection but are 
tailored to Lipschitz functions (apply the Exponential Mechanism~\cite{DBLP:conf/focs/McSherryT07} or the Report-Noise-Max paradigm~\cite{DworkRothBook2014}) and works~\cite{LiuT19-private-select,DBLP:conf/iclr/Papernot022,CLNSS:ITCS2023} that do apply with general private algorithms but significantly depart from the natural one-shot approach: They make a randomized number of computations that is generally much larger than $m$, with each $\calA_i$ invoked multiple times or none. The interpretation of the selection deviates from top-$1$ and does not naturally extend to top-$k$.
\footnote{\cite{LiuT19-private-select} proposed two algorithms for selecting a top candidate from $\eps$-DP candidates.  Their design uses a randomized overall number of applications the candidates (that is, the algorithms $\{\calA_i\}$). Each call is made with $\calA_i$ where $i\sim[m]$ is selected uniformly at random. The first algorithm has overall privacy parameter that is $\approx 2\eps$ and the output has quantile guarantees.  The second algorithm has privacy parameter $3\eps$ and returns the top score over all calls. Due to the randomized invocations, a logarithmic factor increase in the number of calls is needed in order to make sure each algorithm is called at least once.}
We seek privacy analysis that applies to one-shot top-$k$ selection with candidates that are general private algorithms.


The departures made in prior works from the natural interactive paradigm and one-shot selection were essentially compromises:  Simple arguments (that apply with both top-$1$ one-shot private selection~\cite{LiuT19-private-select} and \AboveThreshold tests) show that SVT-like benefits are not possible: 
If we perform $m$ computations that are $\eps$-DP (that is, $m$ candidates or $m$ tests),  the privacy parameter value for a pure DP bound is $\Omega(m)\eps$ and the parameter values for an approximate DP bound are $(\Omega(\eps\log(1/\delta)),\delta)$. This is a daunting overhead -- the privacy charge is of $O(\log(1/\delta))$ instead of $O(1)$ invocations.  
The departure allowed for the appealing benefits of pure-DP and remarkably, for low privacy overhead (factor of 2 or 3 increase in the $\eps$ parameter) even with a single ``above'' response or a single selection. 

We revisit the natural paradigms for interactive accesses and one-shot selection, for their simplicity, interpretability, and generality, with a fresh approach. Considering the mentioned limitations, we take approximate DP to be a reasonable compromise (that is anyhow necessary with advanced composition and other divergences). Additionally, we aim for the regime 
where many private computations are performed on the same dataset and out of these many computations we expect multiple, say $\Omega(\log(1/\delta))$,  ``target hits'' (e.g. positive tests and sum of the $k$-values of selections).  
With these particular relaxations in mind, can we obtain 
SVT-like benefits (e.g. privacy charge that corresponds to $O(1)$ calls per ``target hit'') with the natural paradigm? 
Moreover, can we integrate private top-$k$ selections in a unified target-charging analysis, so that each top-$k$ selection  we perform amounts to $O(k)$ additional target hits?
Such unification would facilitate tighter analysis with advanced composition (performed over all target hits) and amortize overheads.

\section{Overview of Contributions}
We introduce the
{\em Target-Charging Technique (TCT)} for 
privacy analysis over interactive private computations (see Algorithm~\ref{algo:targetcharge}).
Each computation performed on the sensitive dataset $D$ is specified by a private algorithm $\calA_i$ and \emph{target} pairs $\top_i$. The interaction is halted after a pre-specified number $\tau$ of computations that satisfy $\calA_i(D) \in \top_i$.
We define targets as follows:
\begin{definition}[$q$-Target]\label{def:qtarget}
Let $\calM:X^n\to \calY$ be a randomized algorithm. 
For $q\in(0,1]$ and $\eps>0$, we say that a subset $\top \subseteq \calY$ of all possible outcomes is a \emph{$q$-Target} of $\calM$ if the following holds:
For any pair $D^0$ and $D^1$ of neighboring data sets, there exist
$p\in [0,1]$, and three distributions $\mathbf{C}$, $\mathbf{B}^0$ and $\mathbf{B}^1$ such that
\begin{enumerate}
\item The distributions $\calM(D^0)$ and $\calM(D^1)$ can be written as the following mixtures:
\begin{align*}
\calM(D^0) &\equiv p \cdot \mathbf{C} + (1-p) \cdot \mathbf{B}^0,
\\ \calM(D^1) &\equiv p\cdot \mathbf{C} + (1-p)\cdot \mathbf{B}^{1}.
\end{align*}
\item $\mathbf{B}^0,\mathbf{B}^{1}$ are $(\eps,0)$-indistinguishable,
\item $\min(\Pr[\mathbf{B}^0\in \top],\Pr[\mathbf{B}^{1}\in \top]) \ge q$.
\end{enumerate}
\end{definition}

The effectiveness of a target as a proxy of the actual privacy cost is measured by its $q$-value where $q\in (0,1]$. We interpret $1/q$ as the \emph{overhead factor} of the actual privacy exposure per target hit, that is, the number of private accesses that correspond to a single target hit. 
Note that an algorithm with a $q$-target for $\eps>0$ must be 
$(\eps,0)$-DP and that any $(\eps,0)$-DP algorithm has a $1$-target, as the set of all outcomes $\top = \calY$ is a $1$-target (and hence also a $q$-target for any $q\leq 1$). The helpful targets are ``smaller'' (so that we are less likely to be charged) with larger $q$ (so that the overhead per charge is smaller).
We establish the following privacy bounds.

\begin{lemma}[simplified meta privacy cost of target-charging]\label{metaprivacy:lemma}
The privacy parameters of Algorithm~\ref{algo:targetcharge}
(applied with $\eps$-DP  algorithms $\calA_i$ and $q$-targets $\top_i$ until targets are hit $\tau$ times)  is $(\eps',\delta)$ where $\eps' \approx \frac{\tau}{q} \eps$ and $\delta = e^{-O(\tau)}$. 

Alternatively, we obtain parameter values $(\eps',\delta')=(f_\eps(r,\eps),f_\delta(r,\eps)+e^{-O(\tau)})$ where $r\approx \tau/q$ and $(f_\eps(r,\eps),f_\delta(r,\eps))$ are privacy parameter values for advanced composition~\cite{DBLP:conf/focs/DworkRV10}  of $r$ $\eps$-DP computations.
\end{lemma}
The proof is provided in Section~\ref{TC:sec} for a 
precise and more general statement that applies with approximate DP algorithms (in which case the $\delta$ privacy parameter values of all calls add up).  The proof idea is simple but surprisingly powerful: We compare the execution of Algorithm~\ref{algo:bwrapper} on two neighboring data sets $D^0,D^1$. Given a request $(\calA,\top)$, let $p,\mathbf{C},\mathbf{B},\mathbf{B}^0,\mathbf{B}^1$ be the decomposition of $\calA$ w.r.t. $D^0,D^1$ given by Definition~\ref{def:qtarget}. Then, running $\calA$ on $D^0,D^1$ can be implemented in the following equivalent way: we first flip a $p$-biased coin. With probability $p$, the algorithm samples from $\mathbf{C}$ and returns the result. Note that in this case, we do \emph{not} need to access $D^0,D^1$ at all! Otherwise, the algorithm needs to sample from $\mathbf{B}^0$ or $\mathbf{B}^1$, depending on whether the private data is $D^0$ or $D^1$. However, by Property~3 in Definition~\ref{def:qtarget}, there is a decent chance (e.g., with probability at least $q$) that Algorithm~\ref{algo:targetcharge} will ``notice'' the privacy-leaking computation by observing a result in the target set $\top$. If this indeed happens, the algorithm increments the counter. On average, each counter increment corresponds to $\frac{1}{q}$ many accesses to the private data. Finally, when $\tau$ is moderately large we apply a concentration inequality to bound the probability that the actual number of calls much exceeds its expectation of $\tau/q$.

The TCT analysis uses the number of target hits (multiplied by $1/q$) as a proxy for the actual privacy leak, with tail bounds applied to obtain high confidence bounds on the error. The multiplicative error decreases when the number $\tau$ of target hits is larger. In the regime $\tau > \ln(1/\delta)$,  we amortize the mentioned $O(\log(1/\delta))$ overhead of the natural paradigm and achieve  SVT-like bounds where each target hit results in privacy cost equivalent to $O(1/q)$ calls.
In the regime of very few target hits (e.g., few private tests or private selections), we still have to effectively ``pay'' for the larger $\tau =\Omega(\ln(1/\delta))$, but TCT still has some advantages over alternative approaches, due to its use of the natural paradigm and its applicability with general private algorithms. 

TCT can be extended to the case where algorithms have varied privacy parameter and target overhead values that may be adaptively chosen. 
A simple analysis can work with the smallest values encountered. With a tighter analysis, we can expect  
$\sum_i \eps_i/q_i$ to roughly replaces $\tau\eps/q$, but this requires calculation of tighter tail and composition bounds is more complex~\cite{KairouzOV:ICML2015,RogersRUV:NIPS2016} and does not have simple forms. These are useful (but technical) extensions that we leave for follow up work.

\begin{algorithm2e}[h]
\small{
    \caption{Target Charging}
    \label{algo:targetcharge}
    \DontPrintSemicolon
    \KwIn{
         Dataset $D = \{x_1,\dots, x_n\}\in X^n$. Integer $\tau\ge 1$ (Upper limit on the number of target hits). Fraction $q\in[0,1]$. 
    }
    $C \gets 0$\tcp*{Initialize target hit counter}
        \While(\tcp*[f]{Main loop}){$C< \tau$}{
            \textbf{Receive} $(\calA,\top)$ where $\calA$  is an $\eps$-DP mechanism, and $\top$ is a $q$-target for $\calA$ \; 
            $r\gets \calA(D)$\;
           \textbf{Publish} $r$ \;
           \lIf(\tcp*[f]{outcome is a target hit}){$r\in \top$}{$C\gets C+1$}
        } 
        }
\end{algorithm2e}

Despite its simplicity, TCT turns out to be surprisingly powerful due to the 
existence of natural targets with low overhead.  We present an expansive toolkit that is built on top of TCT and describe application scenarios.

\subsection{\texttt{NotPrior} targets} \label{notpriorintro:sec}

A \NotPrior target of an 
$\eps$-DP algorithm is specified by any outcome of our choice (the ``prior") that we  denote by $\bot$. The 
\NotPrior target is the set of all outcomes except $\bot$.  Surprisingly perhaps, 
this is an effective target (See Section~\ref{notprior:sec} for the proof that applies also with approximate-DP):
\begin{lemma} [Property of a \NotPrior target] \label{lemma:NotPriorprivacy}
Let $\calA : X \to \calY \cup \{ \bot\}$, where $\bot\not\in \calY$, be an $\eps$-DP
algorithm.  Then the set of outcomes $\calY$ constitutes an $\frac{1}{e^{\eps}+1}$-target for $\calA$.
\end{lemma}
Note that for small $\eps$, we have $q$ approaching $1/2$ and thus the overhead factor is close to $2$. The TCT privacy analysis is beneficial over plain composition when the majority of all outcomes in our interaction match their prior $\bot$. 
We describe application scenarios for \NotPrior targets. For most of these scenarios, TCT is the only method we are aware of that provides the stated privacy guarantees in the general context. 

\paragraph{Private testing}
A private test is a private algorithm with a Boolean output.
By specifying our prior to be a negative outcome, we obtain an overhead of $2$ (for small $\eps$) for positive responses, which matches the overhead of SVT. TCT is the only method we are aware of that provides SVT-like guarantees with general private tests.

\paragraph{Pay-only-for-change} When we have a prior on the result of each computation and expect the results of most computations to agree with their respective prior, we set $\bot$ to be our prior.  We report all results but pay only for those that disagree with the prior. 
We describe some use cases where paying only for change can be very beneficial  (i) the priors are results of the same computations on an older dataset, so they are likely to remain the same (ii) In streaming or dynamic graph algorithms, the input is a sequence of updates where typically the number of changes to the output is much smaller than the number of updates.  Differential privacy was used to obtain algorithms that are robust to adaptive inputs~\cite{HassidimKMMS20,BKMNSS22} by private aggregation of non-robust copies. The pay-only-for-change allows for number of changes to output (instead of the much larger number of updates) that is quadratic in the number of copies.  Our result enables such gain with any private aggregation algorithm (that is not necessarily in the form of \AboveThreshold tests).


\subsection{Conditional Release}  We have a private algorithm $\calA : X \to \calY$ but are interested in the output $\calA(D)$ only when a certain condition holds (i.e., when the output is in  $\top \subseteq \calY$).
The condition may depend on the interaction transcript thus far (depend on prior computations and outputs). We expect most computations not to meet their release conditions and want to be ``charged'' only for the ones that do. Recall that with differential privacy, not reporting a result also leaks information on the dataset, so this is not straightforward. We define $A_\top := \CR(\calA,\top)$ as the operation that inputs a dataset $D$, computes
$y \gets \calA(D)$. If $y\in \top$, then publish $y$ and otherwise publish $\bot$.
We show that this operation can be analysed in TCT as a call with the 
algorithm and \NotPrior target pair $(\calA_{\top},\top)$, that is, a target hit occurs if and only if $y\in\top$:
\begin{lemma} [\CR privacy analysis] \label{CRprivacy:lemma}
$\calA_\top$ satisfies the privacy parameters of $\calA$ and $\top$ is a \NotPrior target of $\calA_\top$.
\end{lemma}
\begin{proof}
$\calA_\top$ processes the output of the private algorithm $\calA$ and thus from post processing property is also private with the same privacy parameter values. 
Now note that $\top$ is a \NotPrior target of $\calA$, with respect to prior $\bot$.  
\end{proof}

We describe some example use-cases:
\begin{trivlist}
\item (i)
Private learning of models from the data (clustering, regression, average, ML model) but we are interested in the result only when its quality is sufficient, say above a specified threshold, or when some other conditions hold. 
\item (ii)
Greedy coverage or representative selection type applications, where we incur privacy cost only for selected items. To do so, we condition the release on the ``coverage'' of past responses. For example, when greedily selecting a subset of features that are most relevant or a subset of centers that bring most value.
\item (iii)
Approximate \AboveThreshold tests on Lipschitz functions, with release of above-threshold noisy values: As mentioned, SVT incurs additional privacy cost for the reporting whereas TCT (using \CR) does not, so TCT benefits in the regime of sufficiently many target hits.
\item (iv)
\AboveThreshold tests with sketch-based approximate distinct counts:
Distinct counting sketches~\cite{FlajoletMartin85,hyperloglog:2007,Concavesub:KDD2017} meet the privacy requirement by the built-in sketch randomness~\cite{SmithST:NeurIPS2020}. 
We apply
\CR and set $\top$ to be above threshold values.
In comparison, despite the function (distinct count) being 1-Lipschitz, the use of SVT for this task incurs higher overheads in utility (approximation quality) and privacy:
Even for the goal of just testing, a direct use of SVT treats the approximate value as the non-private input, which reduces accuracy due to the additional added noise.  Treating the reported value as a noisy Lipschitz still incurs accuracy loss due to the threshold noise, threshold noise introduces bias, and analysis is complicated by the response not following a particular noise distribution. For releasing values, SVT 
as a separate distinct-count sketch is needed to obtain an independent noisy value~\cite{DBLP:journals/pvldb/LyuSL17}, which increases both storage and privacy costs.

\end{trivlist}

\subsection{Conditional Release with Revisions} \label{condreleaseintro:sec}

We present an extension of Conditional Release that allows for followup \emph{revisions} of the target.
The initial \CR and the followup \RCR calls are described in Algorithm~\ref{algo:conditionalrelease}. 
The \CR call specifies a computation identifier $h$ for later reference, an algorithm and a target pair
$(\calA,\top)$.  It draws $r_h \sim \calA(D)$ and internally stores $r_h$ and a current target $\top_h\gets \top$. 
When $r_h\in \top$ then $r_h$ is published and a charge is made. Otherwise, $\bot$ is published. Each (followup) \RCR call specifies an identifier $h$ and a disjoint extension $\top'$ to its current target $\top_h$. If $r_h\in \top'$, then $r_h$ is published and a charge is made. Otherwise, $\bot$ is published. The stored current target for computation $h$ is augmented to include  $\top'$. Note that a target hit occurs at most once in a sequence of (initial and followup revise) calls and if and only if the result of the initial computation $r_h$ is in the final target $\top_h$.


\begin{algorithm2e}[h]
\small{
    \caption{Conditional Release and Revise Calls}
    \label{algo:conditionalrelease}
    \DontPrintSemicolon
    \tcp{Initial Conditional Release call: Analysed in TCT as a $(\eps,\delta)$-DP algorithm $\calA_{\top}$ and \NotPrior target $\top$}
    \SetKwProg{Fn}{Function}{:}{}
    \Fn(\tcp*[f]{unique identifier $h$, an $(\eps,\delta)$-DP algorithm $\calA\to \calY$, $\top\subset \calY$}){\CR{$h,\calA,\top$}}{ 
    $\top_h\gets \top$\tcp*{Current target for computation $h$}
    \textbf{TCT Charge} for $\delta$ \tcp*{If $\delta>0$, see Section~\ref{TC:sec}}
    $r_h\gets \calA(D)$\tcp*{Result for computation $h$}
      \eIf(\tcp*[f]{publish and charge only if outcome is in $\top_h$}){$r_h\in\top_h$}{ \textbf{Publish} $r_h$ \;
        \textbf{TCT Charge} for a \NotPrior target hit of an $\eps$-DP algorithm
      }
      {\textbf{Publish} $\bot$ \;
    }
 }    
\tcp{Revise call: Analysed in TCT as a $2\eps$-DP Algorithm $(\calA \mid \neg \top_h)_{\top'}$ and \NotPrior target $\top'$}
\Fn(\tcp*[f]{Revise target to include $\top'$}){\RCR{$h,\top'$}}{ 
\KwIn{An identifier $h$ of a prior \CR call, target extension $\top'$ where $\top'\cap  \top_h = \emptyset$}
  \eIf(\tcp*[f]{Result is in current target, publish and charge}){$r_h \in \top'$}
    {\textbf{Publish} $r_h$ \;
    \textbf{TCT Charge} for a \NotPrior target hit of an $2\eps$-DP algorithm
    }
    {\textbf{Publish} $\bot$ \;
    }
    $\top_h \gets \top_h \cup \top'$\tcp*{Update the target to include extension}
}
}
\end{algorithm2e}

We  show the following (Proof provided in Section~\ref{CondRelease:sec}):
\begin{lemma}[Privacy analysis for Algorithm~\ref{algo:conditionalrelease}]\label{lemma:analyze-revise}
Each \RCR call can be analysed in TCT as a call to a $2\eps$-DP algorithm with a \NotPrior target $\top'$. 
\end{lemma}
Thus, the privacy cost of conditional release followed by a sequence of revise calls is within a factor of 2 (due to the doubled privacy parameter on revise calls) of  a single \CR call made with the final target.

The revisions extension of conditional release facilitates our results for private selection, which are highlighted next.

\subsection{Private Top-$k$ Selection} \label{topkselectintro:sec}

One-shot top-$k$ selection is described in Algorithm~\ref{algo:top-k}: We call each algorithm once and report the $k$ responses with highest quality scores.
We establish the following:
\begin{lemma}[Privacy of One-Shot Top-$k$ Selection] \label{oneshotprivacy:lemma}
Consider one-shot  top-$k$ selection (Algorithm~\ref{algo:top-k}) on a dataset $D$ where $\{\calA_i\}$ are $(\eps,\delta_i)$-DP. 
This selection can be simulated exactly in TCT by a sequence of calls to $(2\eps,\delta)$-DP algorithms with \NotPrior targets that has $k$ target hits.

As a corollary, assuming $\varepsilon < 1$, Algorithm~\ref{algo:top-k} is $(O(\eps\sqrt{k\log(1/\delta)}),2^{-\Omega(k)}+\delta+\sum_i\delta_i)$-DP for every $\delta > 0$.
\end{lemma}
To the best of our knowledge, our result is the first such bound for one-shot selection from general private candidates. For the case when the only computation performed on $D$ is a single top-$1$ selection, we match the ``bad example'' in~\cite{LiuT19-private-select} (see Theorem~\ref{thm:vanilla-selection}).
In the regime where $k > \log(1/\delta)$ our bounds generalize those specific to Lipschitz functions in~\cite{DBLP:conf/nips/DurfeeR19,DBLP:conf/icml/QiaoSZ21} (see Section~\ref{sec:selection}). 
Importantly, Lemma~\ref{oneshotprivacy:lemma}
allows for a unified privacy analysis of  
interactive computations that are interleaved with one-shot selections. We obtain 
$O(1)$ overhead per target hit when there are 
$\Omega(\log(1/\delta))$ hits in total.

\begin{algorithm2e}[h]
\small{
    \caption{One-Shot Top-$k$ Selection}
    \label{algo:top-k}
    \DontPrintSemicolon
    \KwIn{
        A dataset $D$. Candidate algorithms $\calA_1,\dots, \calA_m$. Parameter $k\le m$.
    }
    $S \gets \emptyset$ \;
    \For{$i=1,\dots, m$}{
        $(y_i,s_i)\gets \calA_i(D)$ \;
        $S\gets S \cup \{(i,y_i,s_i)\}$
    }
\KwRet{$L\gets$ the top-$k$ triplets from $S$,  by decreasing $s_i$}
}
\end{algorithm2e}

The proofs of Lemma~\ref{oneshotprivacy:lemma} and implications to selection tasks are provided in Section~\ref{sec:selection}. The proof utilizes Conditional Release with revisions (Section~\ref{condreleaseintro:sec}).

\subsubsection{Selection using Conditional Release} \label{appselectionintro:sec}
We analyse private selection tasks using conditional release (see Section~\ref{sec:selection} for details). First note that 
\CR calls (without revising) suffice for \emph{one-shot above-threshold} selection (release all results with quality score that exceeds a pre-specified threshold $t$), with target hits only on what was released: We simply specify the release condition to be $s_i > t$. What is missing in order to implement one-shot top-$k$ selection is an ability to find the ``right'' threshold (a value $t$ so that exactly $k$ candidates have quality scores above $t$), while incurring only $k$ target hits. 
The revise calls provide the functionality of lowering the threshold of previous conditional release calls (lowering the threshold  amounts to augmenting  the  target). This functionality allows us to simulate a sweep of the $m$ results of the batch in the order of decreasing quality scores.  We can stop the sweep when a certain condition is met (the condition must be based on the prefix of the ordered sequence that we viewed so far) and we incur target hits only for the prefix.
To simulate a sweep, we run a high threshold $t$ conditional release of all $m$ candidates and then incrementally lower the threshold $t\gets t-dt$ using sets of $m$ revise calls (one call per candidate).  The released results are in decreasing order of quality scores. 
The one-shot top-$k$ selection (Algorithm~\ref{algo:top-k}) is simulated exactly by such a sweep that stops after $k$ scores are released.  Hence, the same privacy analysis holds and Lemma~\ref{oneshotprivacy:lemma} follows.
We emphasize that the sweeping simulation is only for analysis. The implementation is described in Algorithm~\ref{algo:top-k}.

As mentioned, with this approach we can apply \emph{any stopping condition that depends on the prefix}. This allows us to use data-dependent selection criteria.
One natural such criteria (instead of using a rigid value of $k$) is to choose $k$ when there is a large gap in the quality scores, that the $(k+1)$st quality score is much lower than the $k$th score~\cite{pmlr-v151-zhu22e}.  This criterion can be implemented using  a  one-shot algorithm and  analyzed in the same way using an equivalent sweep.   Data-dependent criteria are also commonly used in applications such as clustering (choose ``the right'' number of clusters according to gap in clustering cost) and greedy selection of representatives.





\subsection{Best of multiple targets} \label{bestmultiintro:sec}

\emph{Multi-target} charging, 
described  in Algorithm~\ref{algo:ktargetcharge}, is a simple but useful extension of Algorithm~\ref{algo:targetcharge} (that is ``single target'').
With $k$-TCT, queries have the form $\big(\calA,(\top_i)_{i\in[k]}\big)$  where $\top_i$ for $i\in[k]$ are $q$-targets (we allow targets to overlap). The algorithm maintains $k$ counters $(C_i)_{i\in[k]}$. For each query, for each $i$, we increment $C_i$ if $r\in \top_i$.   We halt when 
$\min_i C_i = \tau$.  

The multi-target extension allows us to flexibly reduce the total privacy cost to that of the ``best'' among $k$ target indices {\em in retrospect} (the one that is hit the least number of times).  Interestingly, this extension is almost free in terms of privacy cost: The number of targets  $k$ only multiplies the $\delta$ privacy parameter (see Section~\ref{multiTCT:sec} for the proof):
\begin{lemma}[Privacy of multi-TCT]\label{lemma:kTCprivacy}
Algorithm~\ref{algo:ktargetcharge} satisfies $(\eps',k\delta')$-approximate DP bounds, where $(\eps',\delta')$ are privacy bounds for single-target charging (Algorithm~\ref{algo:targetcharge}).
\end{lemma}

Specifically, when we expect that one (index) of multiple outcomes
$\bot_1,\ldots, \bot_k$ will dominate our interaction but can not specify which one it is in advance, we can use $k$-TCT with \NotPrior targets with priors $\bot_1,\ldots,\bot_k$.  From
Lemma~\ref{lemma:kTCprivacy}, the overall privacy cost depends on the number of times that the reported output is different than the most dominant outcome.  
More specifically, for private testing, 
when we expect that one type of outcome would dominate the sequence but we do not know if it is $0$ or $1$, we can apply $2$-TCT. The total number of target hits corresponds to the less dominant outcome. The total number of privacy charges (on average) is at most (approximately for small $\eps$) double that, and therefore 
is always comparable or better to composition (can be vastly lower when there is a dominant outcome). 

\subsection{\texttt{BetweenThresholds} in TCT} \label{betweenintro:sec}

The \BetweenThresholds classifier is a refinement of the \AboveThreshold test. \BetweenThresholds reports if the noisy Lipschitz value is below, between, or above two thresholds $t_l<t_r$. \BetweenThresholds was analysed in \cite{BunSU:SODA2017} in the SVT framework (using noisy thresholds) and it was shown that the overall privacy costs may only depend on the ``between" outcomes. Their analysis required that  $t_r-t_l\geq (12/\eps)(\log(10/\eps)+\log(1/\delta)+1)$. 
We consider the ``natural'' private \BetweenThresholds classifier that compares  the value with added $\Lap(1/\eps)$ noise to the thresholds. We show 
(see Section~\ref{qvotbetween:sec})
that the 
``between'' outcome is a target with $q \geq (1-e^{-(t_r-t_l)\eps})\cdot \frac{1}{e^{\eps}+1}$.
Note that the $q$-value is smaller by a factor of 
$(1-e^{-(t_r-t_l)\eps})$ compared  with
\NotPrior targets. Therefore, there is smooth degradation in the effectiveness of the between outcome as the target as the gap $t_r-t_l$ decreases, and matching  \AboveThreshold when the gap is large. Also note that we require much smaller gaps $t_r-t_l$ compared with~\cite{BunSU:SODA2017}, also asymptotically ($O(\log(1/\eps))$ factor improvement).  Our result brings the use of \BetweenThresholds into the practical regime.

Taking a step back, we compare an \AboveThreshold test with a threshold $t$ with a \BetweenThresholds classifier with $t_l = t -1/\eps$ and $t_r = t+1/\eps$.
Counter-intuitively perhaps, despite \BetweenThresholds being {\em more informative} than \AboveThreshold, as it provides more granular information on the value, its privacy cost is {\em lower} for queries where values are either well above or well below the thresholds (since target hits are unlikely also when queries are well above the threshold). Somehow, the addition of a third outcome to the test allowed for finer privacy analysis! A natural question that arises is whether we can extend this benefit more generally -- inject a ``boundary outcome'' when our private algorithm does not have one, to tighten the privacy analysis. We introduce next a method that achieves this goal.

\subsection{The Boundary Wrapper method} \label{boundarywrapperintro:sec}

 When the algorithm is a tester or a classifier, the result is most meaningful when one outcome dominates the distribution $\calA(D)$.  Moreover, when performing a 
sequence of tests or classification tasks we might expect most queries to have high confidence labels (e.g., \cite{DBLP:conf/iclr/PapernotSMRTE18,BassilyTT:NEURIPS2018}). 
Our hope then is to incur privacy cost that depends only on the ``uncertainty,'' captured by the probability of non-dominant outcomes.

Recall that when we have for each computation a good prior on which outcome is most likely, this goal can be achieved using \NotPrior targets (Section~\ref{notpriorintro:sec}). When we expect the whole sequence to be dominated by one type of outcome, even when we don't know which one it is, this goal can be achieved via \NotPrior with multiple targets (Section~\ref{bestmultiintro:sec}).  But these approaches do not apply when a dominant outcome exists in most computations, but we have no handle on it and it can change arbitrarily between computations in the same sequence. 

For a private test $\calA$, can we somehow choose a moving target {\em per computation} to be the value with the smaller probability 
$\arg\min_{b\in\{0,1\}} \Pr[\calA(D) = b]$?  More generally, with a private classifier, can we somehow choose the target to be all outcomes except for the most likely one?  

Our proposed {\em boundary wrapper}, described in Algorithm~\ref{algo:bwrapper}, is a mechanism that achieves that goal.
The privacy wrapper $\calW$ takes any private algorithm $\calA$, such as a tester or a classifier, and wraps it to obtain algorithm $\calW(\calA)$. The wrapped algorithm has its outcome set augmented to include one \emph{boundary} outcome $\top$ that is designed to be a $q$-target. The wrapper returns $\top$ with some probability that depends on the distribution
of $\calA(D)$ and otherwise returns a sample from $\calA(D)$ (that is, the output we would get when directly applying $\calA$ to $D$).
We then analyse the wrapped algorithm in TCT. 

Note that the probability of the wrapper $\calA$ returning $\top$ is at most $1/3$ and is roughly proportional to the probability of sampling  an outcome other than the most likely from $\calA(D)$. When there is no dominant outcome the $\top$ probability tops at $1/3$.  Also note that
a dominant outcome (has probability $p\in [1/2,1]$ in $\calA(D)$) has probability $p/(2-p)$ to be reported.  This is at least $1/3$ when $p=1/2$ and is close to $1$ when $p$ is close to $1$. 
For the special case of $\calA$ being a private test, there is always a dominant outcome.

A wrapped \AboveThreshold test provides the benefit of \BetweenThresholds discussed in Section~\ref{betweenintro:sec} where we do not pay privacy cost for values that are far from the threshold (on either side). Note that this is achieved in a mechanical way without having to explicitly introduce two thresholds around the given one and defining a different algorithm.

\begin{algorithm2e}[h]
\small{
    \caption{Boundary Wrapper}
    \label{algo:bwrapper}
    \DontPrintSemicolon
    \KwIn{
         Dataset $D = \{x_1,\dots, x_n\}\in X^n$, a private algorithm $\calA$ 
    }
  $r^* \gets \arg\max_{r}\Pr[\calA(D) = r]$  \tcp*{The most likely outcome of $\calA(D)$} 
   $\pi(D) \gets 1 - \Pr[\calA(D) = r^*]$ \tcp*{Probability that $\calA$ does not return the most likely outcome}
   $c \sim \Ber(\min\left\{\frac{1}{3}, \frac{\pi}{1+\pi}\right\})$\tcp*{Coin toss for boundary}
    \leIf(\tcp*[f]{return boundary or value}){$c=1$}{\textbf{Return} $\top$}{\textbf{Return} $\calA(D)$}
    }
\end{algorithm2e}

We establish the following (proofs provided in Section~\ref{sec:boundarywrapper}).  
The wrapped algorithm is nearly as private as the original algorithm:
\begin{lemma} [Privacy of a wrapped algorithm] \label{wrapperprivacy:lemma}
If $\calA$ is $\eps$-DP then Algorithm~\ref{algo:bwrapper} applied to $\calA$ is $t(\eps)$-DP where
$t(\eps)\leq \frac{4}{3}\eps$.
\end{lemma}
The $q$ value of the boundary target of a wrapped algorithm is as follows:
\begin{lemma} [$q$-value of the boundary target] \label{boundaryq:lemma}
The outcome $\top$ of a boundary wrapper (Algorithm~\ref{algo:bwrapper}) of an $\eps$-DP algorithm is 
a $\frac{e^{t(\eps) - 1}}{2(e^{\eps+t(\eps)} - 1)}$-target.
\end{lemma}

For small $\eps$ we obtain $q\approx t(\eps)/(2(\eps +t(\eps))$.
Substituting $t(\eps) = \frac{4}{3}\eps$ we obtain
$q\approx \frac{2}{7}$.
Since the target $\top$ has probability at most $1/3$, this is a small loss of efficiency ($1/6$ factor overhead) compared with composition in the worst case when there are no dominant outcomes.

The Boundary wrapper method can be viewed as a light-weight way to do privacy analysis that pays only for the ``uncertainty'' of the response distribution $\calA(D)$.  
There are more elaborate (and often more complex computationally) methods based on smooth sensitivity (the stability of $\calA(D)$ to changes in $D$)~\cite{NissimRS:STOC2007,DworkLei:STOC2009,pmlr-v30-Guha13}. 


\paragraph*{Probability oracle vs. Blackbox access}
The boundary-wrapper method assumes that the probability of the most dominant outcome in the distribution $\calA(D)$, when it is large enough, is available to the wrapper.  For some algorithms, these values are readily available, for example, the
Exponential Mechanism~\cite{DBLP:conf/focs/McSherryT07} or when applying known noise distributions for \AboveThreshold, \BetweenThresholds, and Report-Noise-Max~\cite{DBLP:conf/nips/DurfeeR19}. In principle, the probability can always be computed (without incurring privacy cost) but sometimes this can be 
inefficient. We propose in  Section~\ref{nooraclewrap:sec} a boundary-wrapping method that only uses blackbox sampling access to the distribution $\calA(D)$. 

At a very high level, we show that one can run an $(\eps,0)$-DP algorithm $\calA$ twice and observe both outcomes. Then, denote by $\calY$ the range of the algorithm $\calA$. We can show that $E=\{ (y, y') : y\ne y' \}\subseteq \calY\times \calY$ is an $\Omega(1)$-target of this procedure. That is, if the analyst observes the same outcome twice, she learns the outcome ``for free''. If the two outcomes are different, the analyst pays $O(\eps)$ of privacy budget, but she will be able to access both outcomes, which is potentially more informative than a single execution of the algorithm.

\subsubsection{Applications to Private Learning using Non-privacy-preserving Models} \label{privatelearningintro:sec}

Promising recent approaches to achieve scalable private learning through training non-private models include
Private Aggregation of Teacher Ensembles (PATE)~\cite{PapernotAEGT:ICLR2017,DBLP:conf/iclr/PapernotSMRTE18}
and Model-Agnostic private learning~\cite{BassilyTT:NEURIPS2018}.

The private dataset $D$ is partitioned into $k$ parts
$D=D_1\sqcup \dots \sqcup D_k$ and a model is trained (non-privately) on each part.  For multi-class classification with $c$ labels, the trained models can be viewed as functions
$\{f_i:\calX\to [c]\}_{i\in[k]}$. Note that changing one sample in $D$ can only change the training set of one of the models. To privately label an example $x$ drawn from a public distribution, we compute the predictions of all the models $\{f_i(x)\}_{i\in[k]}$ and consider the counts
$n_j = \sum_{i\in [k]} \mathbf{1}\{f_i(x)=j\}$ (the number of models that gave label $j$ to example $x$) for $j\in[c]$.
We then privately aggregate to obtain a privacy-preserving label, for example using the Exponential Mechanism~\cite{DBLP:conf/focs/McSherryT07} or Report-Noisy-Max~\cite{DBLP:conf/nips/DurfeeR19,DBLP:conf/icml/QiaoSZ21}.

This setup is used to process queries (label examples) until the privacy budget is exceeded.  In PATE, the new privately-labeled examples are used to train a new \emph{student} model (and $\{f_i\}$ are called \emph{teacher} models). 
In these applications, tight privacy analysis is critical. Composition over all queries is too lossy -- for $O(1)$ privacy, only allows for $O(k^2)$ queries. For tighter analysis, we seek to replace this with $O(k^2)$ ``target hits.'' These works used a combination of methods including SVT, smooth sensitivity, distance-to-instability, and propose-test-release~\cite{DworkLei:STOC2009,pmlr-v30-Guha13}.
We show that the TCT toolkit provides streamlined tighter analysis:
\begin{trivlist}
\item (i)
The works of~\cite{BassilyTT:NEURIPS2018,DBLP:conf/iclr/PapernotSMRTE18} pointed out that if the teacher models are sufficiently accurate, we expect high agreement $n_j \gg k/2$ for the ground truth label $j$ on most queries. These high-agreement examples are also the more useful ones for training the student model.  Moreover, agreement implies stability and the fine-grained privacy cost (when accounted through the mentioned methods) is lower.
We propose the following method that exploits the stability of queries with agreements: Apply the boundary wrapper (Algorithm~\ref{algo:bwrapper}) on top of the Exponential Mechanism. Then use $\top$ as our target. Agreement queries, where
$\max_j n_j \gg k/2$ (or more finely, when $h = \arg\max_j n_j$  and $n_h \gg \max_{j\in [k]\setminus\{h\}} n_j$) are very unlikely to result in target hits.
\item (ii)
If we expect most queries to be either high agreement $\max_j n_j \gg k/2$ or low agreement $\max_j n_j \ll k/2$  and would like to avoid privacy charges also with very low agreement, we can 
apply \AboveThreshold test to $\max_j n_j$. If above, we apply the exponential mechanism. Otherwise, we report ``Low.'' The wrapper applied to the combined algorithm returns a label in $[c]$, ``Low,'' or $\top$.
Note that ``Low'' is a dominant outcome with no-agreement queries (where the actual label is not useful anyway) and a class label in $[c]$ is a dominant outcome with high agreement.  We only pay privacy for weak agreements.
\item (iii)
\cite{DBLP:conf/iclr/PapernotSMRTE18} proposed the use of example selection with PATE, suggesting that examples where the current student model agrees with teachers are less helpful (and thus should not be selected for training to avoid privacy cost). Our proposed use of the wrapper reduces privacy cost in case of any teacher agreement (whether or not the student agrees). We can enhance that: When we are at a stage in the training where most students predictions agree with the teacher, we can use the student prediction as a prior (using \NotPrior targets) to avoid privacy charges when there is agreement (and even still use the training example if we wish).
\end{trivlist}

\subsection{SVT with individual privacy charging} \label{SVTindividualintro:sec}

As a direct application of TCT privacy analysis, we obtain an improved sparse vector technique that supports fine-grained privacy charging for each item in the dataset. 


SVT with individual privacy charging was introduced by
Kaplan et al \cite{DBLP:conf/colt/KaplanMS21}.  The input is a dataset $D\in\calX^n$ and an online sequence of linear queries that are specified by predicate and threshold value
pairs $(f_i,T_i)$. For each query, the algorithms reports noisy \AboveThreshold test results
$\sum_{x\in D}f_i(x) \gtrsim T$. Compared with the standard SVT, which halts after reporting $\tau$ positive responses, SVT with fine-grained charging maintains a separate budget counter $C_x$ for each item $x$. For each query with a positive response, the 
algorithm only charges items that \emph{contribute} to this query (namely, all the $x$'s such that $f_i(x) = 1$). Once an item $x$ contributes to $\tau$ meaningful queries (that is, $C_x=\tau$), it is removed from the data set. This fine-grained privacy charging allows one to obtain better utility with the same privacy budget, as demonstrated by several recent works \cite{DBLP:conf/colt/KaplanMS21,CLNSSS:ICML2022}.

Our improved SVT with individual charging is described in Algorithm~\ref{algo:svt-individual}. 
We establish the following privacy guarantee (see Section~\ref{SVTindividual:sec} for details):
\begin{theorem} [Privacy of Algorithm~\ref{algo:svt-individual}]\label{SVTindividual:thm}
Assume $\eps < 1$. Algorithm~\ref{algo:svt-individual} is $(O(\sqrt{\tau \log(1/\delta)}\eps, 2^{-\Omega(\tau)} + \delta)$-DP for every $\delta \in (0, 1)$.
\end{theorem}

Compared with the prior work \cite{DBLP:conf/colt/KaplanMS21}: Our algorithm uses the ``natural'' approach of adding Laplace noise and comparing, i.e., computing $\hat{f}_i = \left(\sum_{x\in D}f_i(x)\right) + \Lap(1/\eps)$ and 
testing whether $\hat{f}_i\ge T$, whereas \cite{DBLP:conf/colt/KaplanMS21} 
adds two independent Laplace noises. We support publishing the 
approximate sum $\hat{f}_i$ for
``Above-Threshold'' without incurring additional privacy costs.
Moreover, our analysis is significantly simpler (few lines instead of several pages) and for
the same privacy budget, we improve the utility (i.e., the additive error) by a $\log(1/\varepsilon)\sqrt{\log(1/\delta)}$ factor.  Importantly, our improvement aligns the bounds of SVT with individual privacy charging with those of standard SVT, bringing the former into the practical regime.



\begin{algorithm2e}[h]
\small{
    \caption{SVT with Individual Privacy Charging}
    \label{algo:svt-individual}
    \DontPrintSemicolon
    \KwIn{
        Private data set $D\in \mathcal{X}^n$; privacy budget $\tau > 0$; Privacy parameter $\eps>0$.
    }
    \ForEach{$x\in D$}{
        $C_x\gets 0$ \tcp*{Initialize a counter for item $x$} 
    }
    \For(\tcp*[f]{Receive queries}){$i=1,2,\dots, $}{
        \textbf{Receive} a predicate $f_i:\mathcal{X}\to [0,1]$ and threshold $T_i\in \mathbb{R}$ \;
        $\hat{f}_i \gets \left(\sum_{x\in D}f_i(x)\right) + \Lap(1/\eps) $ \tcp*{Add Laplace noise to count}
        \If(\tcp*[f]{Compare with threshold}){$\hat{f}_i \ge T_i$}{
            \textbf{Publish} $\hat{f}_i$ \;
            \ForEach{$x\in D$ such that $f(x) > 0$}{
                $C_x\gets C_x + 1$ \;
                \If{$C_x = \tau$}{
                    Remove $x$ from $D$ \;
                }
            }
        }
        \Else{
            \textbf{Publish} $\perp$\;
        }
    }
    }
\end{algorithm2e}

\subsection*{Conclusion}
We introduced the Target Charging Technique (TCT), a versatile unified privacy analysis framework that is particularly suitable when a sensitive dataset is accessed multiple times via differentially private algorithms.  We provide a toolkit that is suitable for multiple natural scenarios, demonstrate significant improvement over prior work for basic tasks such as private testing and one-shot selection, describe use cases, and list challenges for followup works. TCT is simple with low overhead and we hope will be adopted in practice.

\ignore{
Related work now integrated in intro. 
Convenient related works links:

\BetweenThresholds via SVT \url{https://arxiv.org/pdf/1604.04618.pdf}

\url{https://arxiv.org/pdf/1603.01699.pdf}  

We cite the following earlier.  Not sure if there is a point elaborating more here.
The work on model agnostic private learning. 
\url{https://papers.nips.cc/paper/2018/hash/aa97d584861474f4097cf13ccb5325da-Abstract.html}
Explicit generation of a private model via sample and aggregate copies of a non-private model.  Use of sparse vector on linear queries (aggregate over models).  We can work with any model/analytics that gives private responses, does not have to be a sample and aggregate.  Are there any examples of such?

``distance to instablity'' framework:  Here the Lipschitz function is the distance of the instance from ``stability".  We start with a non private $f$ and consider the Lipschitz function that is its distance from stability (the number of data elements that need to chaמge for the output to change).  We then simply apply sparse vector and do query release.

Liu Talwar 
\url{https://arxiv.org/pdf/1811.07971.pdf}.

Our ITCS
\url{https://arxiv.org/pdf/2211.12063.pdf}

Other related (not cited yet, maybe we should):
\url{https://papers.nips.cc/paper/2020/file/e9bf14a419d77534105016f5ec122d62-Paper.pdf}

}

\newpage
\edith{
\section*{TODOs?}

\begin{itemize}
\item PATE (or Model Agnostic) improvement -- put details section~\ref{privatelearning:sec}.
\end{itemize}
}

 \bibliographystyle{alpha}
\bibliography{mybib}

\newpage
\appendix
\onecolumn

\section{Preliminaries}

\paragraph*{Notation.}  We say that a function $f$ over datasets is $t$-Lipschitz if for any two neighboring datasest $D^0$, $D^1$, it holds that $|f(D^1)-f(D^0)|\leq t$. For two reals $a,b\ge 0$ and $\eps>0$, we write $a\approx_\eps b$ if $e^{-\eps} b \le a \le e^\eps b$.  

For two random variables $X^0,X^1$, we say that they are $\eps$-indistinguishable, denoted $X^0\approx_\eps X^1$, if their max-divergence and symmetric counterpart are both at most $\eps$. That is,  for $b\in \{0,1\}$, $\max_{S\subseteq\supp(X^b)} \ln\left[\frac{\Pr[X^b\in S]}{\Pr[X^{1-b}\in S]} \right] \leq \eps$.

We similarly say that for $\delta>0$, the random variables are $(\eps,\delta)$-indistinguishable, denoted $X^0\approx_{\eps,\delta} X^1$, if for $b\in \{0,1\}$
\[
\max_{S\subseteq\supp(X^{b})} \ln\left[\frac{\Pr[X^b\in S]-\delta}{\Pr[X^{1-b}\in S]} \right] \leq \eps.
\]
For two probability  distributions, 
$\mathcal{B}^0$, $\mathcal{B}^1$
We extend the same notation  and write
$\mathbf{B}^0 \approx_\eps \mathbf{B}^1$ and $\mathbf{B}^0 \approx_{\eps,\delta} \mathbf{B}^1$ when this holds for random variables drawn from the respective distributions.

The following relates $(\eps,0)$ and $(\eps,\delta)$-indistinguishability with $\delta=0$ and $\delta>0$. 
\begin{lemma}\label{relateapprox:lemma}
Let $\mathbf{B}^0$, $\mathbf{B}^1$ be two distributions.
Then $\mathbf{B}^0 \approx_{\eps,\delta} \mathbf{B}^1$ if and only if we can express them as mixtures
\[
\mathbf{B}^b \equiv (1-\delta)\cdot \mathbf{N}^b   + \delta\cdot \mathbf{E}^b\ , 
\]
where $\mathbf{N}^0 \approx_\eps \mathbf{N}^1$.
\end{lemma}


We treat random variables interchangeably as distributions, and in particular, for a randomized algorithms $\calA$ and input $D$ we use $\calA(D)$ to denote both the random variable and the distribution. We say an algorithm $\calA$ is $\eps$-DP (pure differential privacy), if for any two \emph{neighboring} datasets $D$ and $D'$, $\calA(D)\approx_\eps \calA(D')$. Similarly, we say $\calA$ is $(\eps,\delta)$-DP (approximate differential privacy) if for any two neighboring datasets $D,D'$, it holds that
$\calA(D)\approx_{\eps,\delta} \calA(D')$~\cite{DMNS06}. We  refer to $\eps,\delta$ as the \emph{privacy parameters}.

A {\em private test} is a differentially private algorithm with Boolean output (say
in $\{0,1\}$).

\begin{remark}
    The literature in differential privacy uses different definitions of neighboring datasets but in this work the definition and properties are used in a black-box fashion. TCT, and properties in these preliminaries, apply with an abstraction. 
\end{remark}

The following is immediate from Lemma~\ref{relateapprox:lemma}:
\begin{corollary}[Decomposition of an approximate DP Algorithm]\label{approxmix:coro}
An algorithm $\calA$ is $(\eps,\delta)$-DP if and only if
for any two neighboring datasets $D^0$ and $D^1$ we can represent each distribution $\calA(D^b)$ ($b\in\{0,1\}$) as a mixture
\[
\calA(D^b) \equiv (1-\delta)\cdot \mathbf{N}^b   + \delta\cdot \mathbf{E}^b\ , 
\]
where $\mathbf{N}^0 \approx_\eps \mathbf{N}^1$.
\end{corollary}

Differential privacy satisfies the post-processing property (post-processing of the output of a private algorithm remains private with the same parameter values) and also has nice composition theorems:
\begin{lemma}[DP composition~\cite{DMNS06,DBLP:conf/focs/DworkRV10}]\label{composition:lemma}
    An interactive sequence of $r$ executions of  $\eps$-DP algorithms satisfies $(\eps',\delta)$-DP for
    \begin{itemize}
        \item $\eps' = r\eps$ and $\delta=0$  \emph{by basic composition \cite{DMNS06}}, or
        \item 
        for any $\delta>0$,
        \begin{align*}
            \eps' &= \frac{1}{2}r\eps^2 + \eps\sqrt{2r\log(1/\delta)} \ .
        \end{align*}
        \emph{by advanced composition \cite{DBLP:conf/focs/DworkRV10}}.
    \end{itemize}
\end{lemma}

\subsection{Simulation-based privacy analysis}
Privacy analysis of an algorithm $\calA$ via simulations is performed by simulating the original algorithm $\calA$ on two neighboring datasets $D^0,D^1$. The simulator does not know which of the datasets is the actual input (but knows everything about the datasets).  Another entity called the "data holder" has the 1-bit information $b\in \{0,1\}$ on which dataset it is.   We perform privacy analysis with respect to what the holder discloses to the simulator regarding the private bit $b$ (taking the maximum over all choices of $D^0$,$D^1$). The privacy analysis is worst case over the choices of two neighboring datasets.  This is equivalent to performing privacy analysis for $\calA$. 

\begin{lemma}[Simulation-based privacy analysis]\cite{arxiv.2211.06387}\label{lemma:intro-simulate}
Let $\AAA$ be an algorithm whose input is a dataset. If there exist a pair of interactive algorithms $\SSS$ and $H$ satisfying the following 2 properties, then algorithm $\AAA$ is $(\eps,\delta)$-DP.
\begin{enumerate}
    \item For every two neighboring datasets $D^0,D^1$ and for every bit $b\in\{0,1\}$ it holds that $$\left(\SSS(D^0,D^1)\leftrightarrow H(D^0,D^1,b)\right)\equiv\AAA(D^b).$$ Here  
    $\left(\SSS(D^0,D^1)\leftrightarrow H(D^0,D^1,b)\right)$
    denotes the outcome of $\SSS$ after interacting with $H$.
    \item Algorithm $H$ is $(\eps,\delta)$-DP w.r.t.\ the input bit $b$.
\end{enumerate}
\end{lemma}

\subsection{Privacy Analysis with Failure Events} \label{failure:sec}

Privacy analysis of a randomized algorithm $\calA$ using designated failure events is as follows:
\begin{enumerate}
    \item Designate some runs of the algorithm as \emph{failure events}. 
    \item
    Compute an upper bound on the maximum probability, over datasets $D$, of a transcript with a failure designation.
    \item
    Analyse the privacy of the interaction transcript conditioned on no failure designation. 
\end{enumerate}
Note that the failure designation is only used for the purpose of analysis.  The output on failure runs is not restricted (e.g., could be the dataset $D$)

\begin{lemma} [Privacy analysis with privacy failure events] \label{failureevent:lemma}
Consider privacy analysis of $\calA$ with failure events.
If the probability of a failure event is bounded by $\delta^*\in [0,1]$ and the transcript conditioned on non-failure is $(\eps',\delta')$-DP then the algorithm $\calA$ is $(\eps,\delta+\delta^*)$-DP.
\end{lemma}
\begin{proof}
Let $D^0$ and $D^1$ be neighboring datasets.
From our assumptions, for $b\in\{0,1\}$, we can represent $\calA(D^b)$ as the mixture
  $\calA(D^b)\equiv(1-\delta^{b})\cdot \mathbf{Z}^{b}+ \delta^{b}\cdot \mathbf{F}^{b}$, where
  $\mathbf{Z}^0 \approx_{\eps',\delta'} \mathbf{Z}^1$, and $\delta^{(b)} \leq \delta^*$.
From Lemma~\ref{relateapprox:lemma}, we have
$\mathbf{Z}^b \equiv (1-\delta') \cdot \mathbf{N}^b + \delta'\cdot \mathbf{E}^b$, where
$\mathbf{N}^0 \approx_{\eps'} \mathbf{N}^1$.

  Then 
  \begin{align*}
      \calA(D^b) &= (1-\delta^{(b)})\cdot \mathbf{Z}^{b}+ \delta^{(b)}\cdot \mathbf{F}^{(b)} \\
      &= (1-\delta^*)\cdot \mathbf{Z}^{b} + (\delta^*-\delta^{(b)}) \cdot \mathbf{Z}^{b} + \delta^{(b)}
      \cdot \mathbf{F}^{b}\\
      &= (1-\delta^*)\cdot \mathbf{Z}^{b} + \delta^*\cdot \left((1-\delta^{(b)}/\delta^*)\cdot \mathbf{Z}^{b} +  \delta^{(b)}\cdot \mathbf{F}^{b} \right)\\
      &= (1-\delta^*)(1-\delta')\cdot \mathbf{N}^{b} + (1-\delta^*)\delta'\cdot \mathbf{E}^b + \delta^*\cdot \left((1-\delta^{(b)}/\delta^*)\cdot \mathbf{Z}^{b} +  \delta^{(b)}\cdot \mathbf{F}^{b} \right) \\
      &= (1-\delta^*-\delta') \cdot \mathbf{N}^{b} + \delta'\delta^*\cdot \mathbf{N}+ (1-\delta^*)\delta'\cdot \mathbf{E}^b + \delta^*\cdot \left((1-\delta^{(b)}/\delta^*)\cdot \mathbf{Z}^{b} +  \delta^{(b)}\cdot \mathbf{F}^{b} \right)
  \end{align*}
  The claim follows from Corollary~\ref{approxmix:coro}.
\end{proof}


Using simulation-based privacy analysis we can 
treat an interactive sequence of approximate-DP algorithms (optionally with designated failure events)
as a respective interactive sequence of pure-DP algorithms where the $\delta$ parameters are anlaysed through failure events.  This simplifies analysis:

We can relate the privacy of a composition of approximate-DP algorithms to that of a composition of corresponding pure-DP algorithms:
\begin{corollary} [Composition of approximate-DP algorithms] \label{approxcomp:coro}
An interactive sequence of
$(\eps_i,\delta_i)$-DP algorithms ($i\in [k]$) has privacy parameter values $(\eps',\delta'+ \sum_{i=1}^k \delta_i)$, where $(\eps',\delta')$ are privacy parameter values of
a composition of pure $(\eps_i,0)$-DP algorithms $i\in [k]$.
%
\end{corollary}
\begin{proof}
We perform simulation-based analysis.  Fix two neighboring datasets $D^0$, $D^1$.
For an $(\eps_i,\delta_i)$-DP algorithm, we can consider the mixtures as in Corollary~\ref{approxmix:coro}.
We draw $c\sim\Ber(\delta_i)$ and if $c=1$ designate the output as failure and return $r\sim \mathbf{E}^{(b)}$.  Otherwise, we return $r\sim \mathbf{N}^{(b)}$.
The overall failure probability is bounded by 
$1-\prod_i (1-\delta_i) \leq \sum_i \delta_i$.
The output conditioned on non-failure is a composition of
$(\eps_i,0)$-DP algorithms ($i\in [k]$).
The claim follows using Lemma~\ref{failureevent:lemma}.
\end{proof}

\section{The Target-Charging Technique} \label{TC:sec}

We extend the definition of $q$-targets (Definition~\ref{def:qtarget}) so that it applies with approximate DP algorithms:
\begin{definition}[$q$-target with $(\eps,\delta)$ of a pair of distributions] \label{def:approxqtargetdist}
Let $\calA\to\calY$ be a randomized algorithm. 
Let $\mathbf{Z}^0$ and $\mathbf{Z}^1$ be two distributions with support $\calY$.  
We say that $\top\subseteq \calY$ is a $q$-target of 
$(\mathbf{Z}^0,\mathbf{Z}^1)$ with $(\eps,\delta)$, where $\eps>0$ and $\delta\in [0,1)$, 
if there exist 
$p\in [0,1]$ and five distributions $\mathbf{C}$, $\mathbf{B}^b$, and $\mathbf{E}^b$ (for $b\in\{0,1\}$) such that $\mathbf{Z}^0$ and $\mathbf{Z}^1$ can be written as the mixtures
\[
\begin{aligned}
\mathbf{Z}^0 &\equiv (1-\delta)\cdot ( p \cdot \mathbf{C} + (1-p) \cdot \mathbf{B}^0 )+\delta \cdot \mathbf{E}^0\\
\mathbf{Z}^1 &\equiv (1-\delta)\cdot ( p \cdot \mathbf{C} + (1-p) \cdot \mathbf{B}^1 )+\delta \cdot \mathbf{E}^1
\end{aligned}
\]
where
$\mathbf{B}^0 \approx_\eps \mathbf{B}^1$, and
$\min(\Pr[\mathbf{B}^0\in\top],\Pr[\mathbf{B}^1\in\top])\geq q$.
\end{definition}

\begin{definition}[$q$-target with $(\eps,\delta)$ of a randomized algorithm] \label{def:approxqtarget}
Let $\calA\to\calY$ be a randomized algorithm. 
We say that $\top\subseteq \calY$ is a $q$-target of 
$\calA$ with $(\eps,\delta)$, where $\eps>0$ and $\delta\in [0,1)$,
if for any pair $D^0$, $D^1$ of neighboring datasets,
$\top$ is a $q$-target with $(\eps,\delta)$ of
$\calA(D^0)$ and $\calA(D^1)$.
\end{definition}

We can relate privacy of an algorithms or indistinguishability of two distributions to existence of $q$-targets: 
\begin{lemma}
(i)~If $(\mathbf{Z}^0, \mathbf{Z}^1)$ have a $q$-target with $(\eps,\delta)$ then
$\mathbf{Z}^0 \approx_{\eps,\delta} \mathbf{Z}^1$.
Conversely,
if $\mathbf{Z}^0\approx_{\eps,\delta} \mathbf{Z}^1$ then 
$(\mathbf{Z}^0, \mathbf{Z}^1)$ have a $1$-target with $(\eps,\delta)$ (the full support is a $1$-target).

(ii)~If an algorithm $\calA$ has a $q$-target with $(\eps,\delta)$ then $\calA$ is $(\eps,\delta)$-DP. Conversely, if an algorithm $\calA$ is $(\eps,\delta)$-DP then it has a 1-target (the set $\calY$) with $(\eps,\delta)$.
\end{lemma}
\begin{proof}
If two distributions $\mathbf{B}^0$, $\mathbf{B}^1$ have a $q$-target with $(\eps,\delta)$ than from Definition~\ref{def:approxqtargetdist} they can be represented as mixtures. Now observe the if $\mathbf{B}^0\approx_\eps \mathbf{B}^1$ then
the mixtures also satisfy 
$p\cdot \mathbf{C}+(1-p)\cdot \mathbf{B^0} \approx_\eps p\cdot \mathbf{C}+(1-p)\cdot \mathbf{B^0}$. 
Using Lemma~\ref{relateapprox:lemma}, we get 
$\mathbf{Z}^0 \approx_{\eps,\delta} \mathbf{Z}^1$.

For (ii) consider $\calA$ and two neighboring datasets $D^0$ and $D^1$. Using Definition~\ref{def:approxqtarget} and applying the argument above we obtain
$\calA(D^0) \approx_{\eps,\delta} \calA(D^1)$. The claim follows using Corollary~\ref{approxmix:coro}.

Now for the converse. If $\mathbf{Z}^0 \approx_{\eps,\delta} \mathbf{Z}^1$ then consider the decomposition as in 
Lemma~\ref{relateapprox:lemma}.  Now we set $p=0$ and $\mathbf{B}^b \gets \mathbf{N}^b$ to obtain the claim with $q=1$ and the target being the full support.

For (ii), if $\calA \to \calY$ is $(\eps,\delta)$-DP then consider neighboring $\{D^0,D^1\}$. We have
$\calA(D^0) \approx_{\eps,\delta} \calA(D^1)$. We proceed as with the distributions.
\end{proof}

Algorithm~\ref{algo:targetchargeapprox} is an extension of 
Algorithm~\ref{algo:targetcharge} that permits calls to approximate DP algorithms.  The extension also inputs a bound $\tau$ on the number of target hits and a bound $\tau_\delta$ on the cummulative $\delta$ parameter values of the algorithms that were called.
We apply adaptively a sequence of $(\eps,\delta)$-DP algorithms with specified $q$-targets to the input data set $D$ and publish the results.  We halt when the first of the following happens (1) the respective target sets are hit for a specified $\tau$ number of times (2) the accumulated $\delta$-values exceed the specified limit $\tau_\delta$.

\begin{algorithm2e}[h]
\small{
    \caption{Target Charging with Approximate DP}
    \label{algo:targetchargeapprox}
    \DontPrintSemicolon
    \KwIn{
         Dataset $D = \{x_1,\dots, x_n\}\in X^n$. Integer $\tau\ge 1$ (Upper limit on the number of target hits). $\tau_\delta \geq 0$ (upper limit on cumulative $\delta$ parameter). Fraction $q\in[0,1]$. 
    }
    $C \gets 0$, $C_\delta \gets 0$ \tcp*{Initialize target hit and failure counters}
        \For(\tcp*[f]{Main loop}){$i=1,\ldots$}{
            \textbf{Receive} $(\calA_i,\top_i)$ where $\calA_i$  is an $(\eps,\delta_i)$-DP mechanism, and $\top_i$ is a $q$-target with $(\eps,\delta_i)$ for $\calA$\; 
            $r\gets \calA_i(D)$\;
            \lIf{$C_\delta + \delta_i > \tau_\delta$}{\bf{Halt}}
            $C_\delta \gets C_\delta + \delta$\tcp*{TCT charge for $\delta_i$}
           \textbf{Publish} $r$ \;
           \If(\tcp*[f]{TCT Charge for a $q$-target hit with $\eps$}){$r\in \top$}{$C\gets C+1$\; \lIf{$C=\tau$}{\bf{Halt}} }
        } 
        }
\end{algorithm2e}

\begin{algorithm2e}[h]
\small{
    \caption{Simulation of Target Charging}
    \label{algo:simultargetcharge}
    \DontPrintSemicolon
    \KwIn{
         Two neighboring datasets $D^0$, $D^1$, private $b\in \{0,1\}$, $\tau\in \mathbb{N}$, $\tau_\delta \in \mathbb{R}_{\geq 0}$,  $q\in[0,1]$, $\alpha>0$. 
    }
    $C \gets 0$, $C_\delta \gets 0$, $h\gets 0$ \tcp*{Initialize; $h$ is a counter on the number of non-fail calls to data holder}
        \For(\tcp*[f]{Main loop}){$i=1,\ldots$}{
            \textbf{Receive} $(\calA_i,\top_i)$ where $\calA_i$  is an $(\eps,\delta_i)$-DP mechanism, and $\top_i$ is a $q$-target with $(\eps,\delta_i)$ for $\calA$\; 
            \lIf{$C_\delta + \delta_i > \tau_\delta$}{\textbf{Halt}}
            $C_\delta \gets C_\delta + \delta$\;
            Let $p\in [0,1]$, $\mathbf{C}$, $\mathbf{B}^0\approx_\eps \mathbf{B}^1$, and $\mathbf{E}^b$ (for $b\in\{0,1\}$) such that
            $\calA(D^b) \equiv (1-\delta)\cdot ( p \cdot \mathbf{C} + (1-p) \cdot \mathbf{B}^b )+\delta \cdot \mathbf{E}^b$\tcp*{By Definition~\ref{def:approxqtarget}}
            \eIf(\tcp*[f]{Non-private Data Holder call with Failure}){$\Ber(\delta)\equiv 1$}{\textbf{Fail}\; \textbf{Publish} $r\sim \mathbf{E}^b$ }
            {
                \eIf{$\Ber(p)\equiv 1$}{\textbf{Publish} $r\sim \mathbf{C}$ \tcp*{No access to data holder}}
            {\textbf{Publish} $r\sim \mathbf{B}^b$ \tcp*{$\eps$-DP Data Holder Call}
            $h\gets h+1$\tcp*{counter of $\eps$-private data holder calls}
            \If(\tcp*[f]{Number of Holder calls exceeded limit}){$h> (1+\alpha)\tau/q$}{\textbf{Fail}}
            \If(\tcp*[f]{outcome is a target hit}){$r\in \top$}{$C\gets C+1$\; \lIf{$C=\tau$}{\bf{Halt}} }
            }
            }           
        }
        }
\end{algorithm2e}

The privacy cost of Target-Charging is as follows (This is a precise and more general statement of Lemma~\ref{metaprivacy:lemma}):
\begin{theorem}[Privacy of Target-Charging]\label{thm:TCprivacy}
Algorithm~\ref{algo:targetchargeapprox} satisfies the following approximate DP privacy bounds:
\begin{align*}
&\left( (1+\alpha)\frac{\tau}{q}\eps, C_\delta+ \delta^*(\tau,\alpha)\right) , & \text{for any $\alpha>0$;}\\
&\left( \frac{1}{2}(1+\alpha)\frac{\tau}{q} \eps^2  + \eps \sqrt{(1+\alpha)\frac{\tau}{q} \log(1/\delta)}, \delta + C_\delta+ \delta^*(\tau,\alpha) \right), & \text{for any $\delta>0$, $\alpha>0$.}
\end{align*}
where $\delta^*(\tau,\alpha) \leq e^{-\frac{\alpha^2}{2(1+\alpha)} \tau}$ and $C_\delta \leq \tau_\delta$ is as computed by the algorithm. 
\end{theorem}
\begin{proof}
We apply the simulation-based privacy analysis in Lemma~\ref{lemma:intro-simulate} and use privacy analysis with
failure events (Lemma~\ref{failureevent:lemma}).

The simulation is described in Algorithm~\ref{algo:simultargetcharge}.
Fix two neighboring data sets $D^0$ and $D^1$.
The simulator initializes the target hit counter $C\gets 0$ and the cumulative $\delta$-values tracker $C_\delta \gets 0$. 
For $i\geq 1$ it proceeds as follows. It receives  $(\calA_i,\top_i)$ where $\calA_i$ is $(\eps,\delta_i)$-DP. 
If $C_\delta+\delta_i > \tau_\delta$ it halts.
Since $\top_i$ is a $q$-target for $\calA_i$, there are 
$p$, $\mathbf{C}$, $\mathbf{B}^0$, $\mathbf{B}^1$, $\mathbf{E}^0$ and $\mathbf{E}^1$ as in Definition~\ref{def:approxqtarget}.
The simulator flips a biased coin $c' \sim \Ber(\delta)$.
If $c' = 1$ it outputs $r\sim \mathbf{E}^b$ and the execution is designated as \textbf{Fail}. In this case there is an interaction with the data holder but also a failure designation.
The simulator flips a biased coin $c\sim \Ber(p)$. If $c = 1$, then the simulator publishes a sample $r\sim \mathbf{C}$ (this does not require an interaction with the data holder). 
Otherwise, the data holder is called.  The data holder publishes $r \sim \mathbf{B^b}$. We track the number $h$ of calls to the data holder. If $h$ exceeds $(1+\alpha)\tau/q$, we designate the execution as \textbf{Fail}. 
If $r\in \top_i$ then $C$ is incremented. If $C=\tau$, the algorithm halts.

The correctness of the simulation (faithfully simulating Algorithm~\ref{algo:targetcharge} on the dataset $D^b$) is straightforward. We analyse the privacy cost.
We will show that 
\begin{itemize}
    \item [(i)] the simulation designated a failure with probability at most $C_\delta+ \delta^*(\tau,\alpha)$. 
    \item [(ii)]
    Conditioned on no failure designation, the simulation performed at most $r= (1+\alpha) \frac{\tau}{q}$ adaptive calls to 
    $(\eps,0)$-DP algorithms
\end{itemize}
Observe that (ii) is immediate from the simulation declaring failure when $h>r$.  We will establish (i) below.

The statement of the Theorem follows from Lemma~\ref{failureevent:lemma} and when applying the
DP composition bounds (Lemma~\ref{composition:lemma}).
The first bounds follow using basic composition and the second follow using advanced composition \cite{DBLP:conf/focs/DworkRV10}.

This analysis yields the claimed privacy bounds with respect to the private bit $b$.  From Lemma~\ref{lemma:intro-simulate} this is the privacy cost of the algorithm.




It remains to show bound the failure probability.  There are two ways in which a failure can occur.  The first is on each call, with probability $\delta_i$.  This probability is bounded by 
$1-\prod_i \delta_i \leq \sum_i \delta_i \leq C_\delta$.  The second is when the number $h$ of private accesses to the data holder exceeds the limit.  We show that the probability that the algorithm halts with failure due to that is at most $\delta^*$.

We consider a process that continues until $\tau$ charges are made. 
The privacy cost of the simulation (with respect to the private bit $b$) depends on the number of times that the data holder is called. 
Let $X$ be the random variable that is the number of calls to the data holder.
Each call is $\eps$-DP with respect to the private $b$.
In each call, there is probability at least $q$ for a ``charge''  (increment of $C$).

A failure is the event 
that the number of calls to data holder exceeds $(1+\alpha) \tau/q$ before $\tau$ charges are made. We show that this occurs with 
probability at most $\delta^*(\tau,\alpha)$:
 \begin{equation}\label{boundeq}
\Pr\left[X > (1+\alpha) \frac{\tau}{q}\right] \leq \delta^*(\tau,\alpha) \ .
\end{equation}

To establish~\eqref{boundeq}, we first observe that the distribution of the random variable $X$ is dominated by a random variable $X'$ that corresponds to a process of drawing i.i.d.\ $\Ber(q)$ until we get $\tau$ successes
(Domination means that for all $m$, $\Pr[X'>m] \geq \Pr[X>m]$).
Therefore, it suffices to establish that
\[
\Pr\left[X'> (1+\alpha) \frac{\tau}{q} \right] \leq \delta^*(\tau,\alpha) \ .
\]

Let $Y$ be the random variable that is a sum of $m= 1+ \left\lfloor (1+\alpha) \frac{\tau}{q}\right\rfloor$ i.i.d.\ $\Ber(q)$ random variables. 
Note that 
\[
\Pr\left[X'> (1+\alpha) \frac{\tau}{q}\right] = \Pr[Y<\tau]\ .
\]

We bound $\Pr[Y<\tau]$ using
multiplicative Chernoff bounds~\cite{Chernoff52}\footnote{Bound can be tightened when using precise tail probability values.}. 
The expectation is $\mu = mq$ and we bound the probability that the sum of Bernoulli random variables is below $\frac{1}{1+\alpha} \mu = (1- \frac{\alpha}{1+\alpha})\mu$.
Using the simpler form of the bounds we get using $\mu=mq \geq (1+\alpha)\tau$
\[
\Pr[Y<\tau]= \Pr[Y< (1- \frac{\alpha}{1+\alpha})\mu]\leq e^{-\frac{\alpha^2}{2 (1+\alpha)^2} \mu} \leq   e^{-\frac{\alpha^2}{2(1+\alpha)} \tau}\ .
\]
\end{proof}

\begin{remark} [Number of target hits]
The TCT privacy analysis has a tradeoff between the final ``$\eps$'' and ``$\delta$'' privacy parameters.  
There is multiplicative factor of $(1+\alpha)$ ($\sqrt{1+\alpha}$ with advanced composition) on the ``$\eps$'' privacy parameter. But when we use a smaller $\alpha$ we need a larger value of $\tau$ to keep the ``$\delta$'' privacy parameter small.
For a given $\alpha,\delta^*>0$, we can calculate a bound on the smallest value of $\tau$ that works. 
We get
\begin{align*}
    \tau &\geq 2 \frac{1+\alpha}{\alpha^2}\cdot \ln(1/\delta^*) & \text{(simplified Chernoff)} \\
    \tau &\geq  \frac{1}{(1+\alpha)\ln\left(e^{\alpha/(1+\alpha)}(1+\alpha)^{-1/(1+\alpha)}\right) }\cdot \ln(1/\delta^*) & \text{(raw Chernoff)}
\end{align*}
 For $\alpha=0.5$ we get
    $\tau > 10.6 \cdot \ln(1/\delta^*)$.  For $\alpha=1$ we get
    $\tau > 3.26 \cdot \ln(1/\delta^*)$. For $\alpha=5$ we get
    $\tau > 0.31 \cdot \ln(1/\delta^*)$.
    
\end{remark}

\ignore{
\edith{I commented out the below, as I am not sure it adds clarity}
\begin{remark}[Utility and selection of $\top$ and $q$]
The privacy cost of Target-Charging scales with $\tau/q$ (or $\sqrt{\tau/q}$ with advanced composition). Hence, for a given target set $\top$ of outcomes, we seek $q$ that is as large as possible. 
The utility of Target-Charging is higher when we can perform more computations before the budget $\tau$ is exceeded.  
The number of computations depends inversely on the 
``typical'' (over algorithms and data sets) probability that an outcome lies in the target $\Pr[\calM(D)\in \top]$.  Therefore, for higher utility (with a given $q$) we seek smaller targets.
To summarize, to maximize utility (the number of $\eps$-DP computations we perform) for a given privacy budget we seek to choose $q$ and $\top$ as to minimize $\Pr[\calM(D)\in \top]/q$.
\end{remark}
}

\begin{remark} [Mix-and-match TCT]
TCT analysis can be extended to the case where we use algorithms with varied privacy guarantees $\eps_i$ and varied $q_i$ values.\footnote{One of our applications of revise calls to conditional release (see Section~\ref{CondRelease:sec} applies TCT with both $\eps$-DP and $2\eps$-DP algorithms even for base $\eps$-DP algorithm)} In this case the privacy cost depends on
$\sum_{i \mid \calA_i(D) \in \top_i} \frac{\eps_i}{q_i}$. The analysis 
relies on tail bounds on the sum of random variables, is more complex.  Varied $\eps$ values means the random variables have different size supports. A simple coarse bound is according to the largest support, which allows us to use a simple counter for target hits, but may be lossy with respect to precise bounds. The discussion concerns the (analytical or numerical) derivation of tail bounds is non-specific to TCT and is tangential to our contribution.
\end{remark}

\subsection{Multi-Target TCT}  \label{multiTCT:sec}

\begin{algorithm2e}[H]
    \caption{Multi-Target Charging}
    \label{algo:ktargetcharge}
    \DontPrintSemicolon
    \KwIn{
         Dataset $D = \{x_1,\dots, x_n\}\in X^n$. Integer $\tau\ge 1$ (charging limit). Fraction $q\in[0,1]$, $k\geq 1$ (number of targets). 
    }
    \lFor(\tcp*[f]{Initialize charge counters}) {$i\in [k]$}{$C_i \gets 0$}
        \While(\tcp*[f]{Main loop}){$\min_{i\in [k]}C_i< \tau$}{
            {\bf Receive} $(\calA,(\top_i)_{i\in[k]})$ where $\calA$  is an $\eps$-DP mechanism, and $\top_i$ is a $q$-target for $\calA$\;
            $r\gets \calA(D)$\;
           \textbf{Publish} $r$ \;
           \For{$i\in[k]$}{
           \lIf(\tcp*[f]{outcome is in $q$-target $\top_i$}){$r\in \top_i$}{$C_i\gets C_i+1$}}
        }        
\end{algorithm2e}

\begin{proof}[Proof of Lemma~\ref{lemma:kTCprivacy} (Privacy of multi-Target TCT)]
\footnote{We note that the claim generally holds for online privacy analysis with the best of multiple methods. We provide a proof specific to multi-target charging below.}Let $(\eps,\delta)$ be the privacy bounds for $\calM_i$ that is single-target TCT with $(\calA_i,\top_i)$.
Let $\calM$ be the $k$-target algorithm.
Let $\top^j_i$ be the $i$th target in step $j$.

We say that an outcome sequence $R=(r_j)_{j=1}^h\in R$ is 
valid for $i\in [k]$ if and only if
$\calM_i$ would halt with this output sequence, that is, 
$\sum_{j=1}^h \mathbf{1}\{r_j\in \top^j_i\} = \tau$ and $r_h \in \top^h_i$.
We define $G(R)\subset [k]$ to be all $i\in [k]$ for which $R$ is valid.

Consider a set of sequences $H$. Partition $H$ into $k+1$ sets $H_i$ so that $H_0 = \{R\in H \mid G(R)=\emptyset\}$ and
$H_i$ may only include $R\in H$ for which $i\in G(R)$.
That is, $H_0$ contains all sequences that are not valid for any $i$ and $H_i$ may contain only sequences that are valid for $i$.
\begin{align*}
\Pr[\calM(D)\in H] &= \sum_{i=1}^k \Pr[\calM(D)\in H_i] = \sum_{i=1}^k \Pr[\calM_i(D)\in H_i] \\
&\leq \sum_{i=1}^k \big(e^\eps \cdot \Pr[\calM_i(D')\in H_i] + \delta\big) =
e^\eps \cdot \sum_{i=1}^k \Pr[\calM_i(D')\in H_i] + k\cdot \delta \\ &=
e^\eps \Pr[\calM(D')\in H] + k\cdot\delta.
\end{align*}
\end{proof}

\section{Properties of \texttt{NotPrior} targets} \label{notprior:sec}

Recall that a \NotPrior target of an 
$(\eps,\delta)$-DP algorithm is specified by any potential outcome (of our choice) that we  denote by $\bot$. The 
\NotPrior target is the set of all outcomes except $\bot$.  
In this Section we prove (a more general statement of) Lemma~\ref{lemma:NotPriorprivacy}:
\begin{lemma} [Property of a \NotPrior target] 
Let $\calM : X \to \calY \cup \{ \bot\}$, where $\bot\not\in \calY$, be an $(\eps,\delta)$-DP
algorithm.  Then the set of outcomes $\calY$ constitutes an 
$\frac{1}{e^{\eps}+1}$-target with $(\eps,\delta)$ for $\calM$.
\end{lemma}

We will use the following lemma:
\begin{lemma} \label{lemma:limitednotprior}
    If two distributions $\mathbf{Z}^0$, $\mathbf{Z}^1$ with support
$\calY \cup \{ \bot\}$ satisfy $\mathbf{Z}^0 \approx_\eps \mathbf{Z}^1$ then $\calY$ constitutes an 
$\frac{1}{e^{\eps}+1}$-target with $(\eps,0)$ for 
$(\mathbf{Z}^0,\mathbf{Z}^1)$.
\end{lemma}

\begin{proof}[Proof of Lemma~\ref{lemma:NotPriorprivacy}]
From Definition~\ref{def:approxqtarget}, it suffices to show that
for any two neighboring datasets,  $D^0$ and $D^1$, 
the set $\calY$ is an 
$\frac{1}{e^{\eps}+1}$-target with $(\eps,\delta)$ for
$(\calM(D^0),\calM(D^1))$ (as in Definition~\ref{def:approxqtargetdist}).

Consider two neighboring datasets. We have
$\calM(D^0) \approx_{\eps,\delta} \calM(D^1)$.
Using Lemma~\ref{relateapprox:lemma}, for $b\in \{0,1\}$ we can
have
\begin{equation}\label{eq:dec}
    \calM(D^b) = (1-\delta)\cdot \mathbf{N}^b + \delta\cdot \mathbf{E}^b,
\end{equation}
 where $\mathbf{N}^0 \approx_\eps \mathbf{N}^1$.
 From Lemma~\ref{lemma:limitednotprior},
 $\calY$ is a $\frac{1}{e^{\eps}+1}$-target with $(\eps,0)$ for
 $(\mathbf{N}^0,\mathbf{N}^1)$.
 From Definition~\ref{def:approxqtargetdist} and \eqref{eq:dec}, this means that $\calY$ is a $\frac{1}{e^{\eps}+1}$-target with $(\eps,\delta)$ for 
$(\calM(D^0),\calM(D^1))$.
\end{proof}


\subsection{Proof of Lemma~\ref{lemma:limitednotprior}}

We first prove Lemma~\ref{lemma:limitednotprior} for the special case of private testing (when the support is $\{0,1\}$):
\begin{lemma} [target for private testing] \label{lemma:binaryFnP}
Let $\mathbf{Z}^0$ and $\mathbf{Z}^1$ with support $\{0,1\}$ 
satisfy $\mathbf{Z}^0 \approx_\eps \mathbf{Z}^1$
Then $\top=\{1\}$ (or $\top=\{0\}$) is an  $\frac{1}{e^{\eps}+1}$-target with $(\eps,0)$ for $(\mathbf{Z}^0,\mathbf{Z}^1)$.
\end{lemma}
\begin{proof}
We show that Definition~\ref{def:approxqtargetdist} is satisfied with 
$\top=\{1\}$, $q=\frac{1}{e^{\eps}+1}$ and $(\eps,0)$, and $\mathbf{Z}^0,\mathbf{Z}^1$.
\begin{align*}
\pi &= \Pr[\mathbf{Z}^0 \in \top] \\
\pi' &= \Pr[\mathbf{Z}^1 \in \top]
\end{align*}
be the probabilities of $\top$ outcome in $\mathbf{Z}^0$ and 
$\mathbf{Z}^1$ respectively. Assume without loss of generality (otherwise we switch the roles of $\mathbf{Z}^0$ and $\mathbf{Z}^1$) that $\pi' \geq \pi$.
If $\pi \geq \frac{1}{e^{\eps}+1}$, the choice of $p=0$ and
$\mathbf{B}^b = \mathbf{Z}^b$ (and any $\mathbf{C}$) trivially satisfies the conditions of Definition~\ref{def:qtarget}.
Generally, (also for all $\pi < \frac{1}{e^{\eps}+1}$):
\begin{itemize}
\item
Let
\[
p=1- \frac{\pi'e^\eps-\pi}{e^\eps-1}\ .
\]
Note that since $\mathbf{Z}^0 \approx_\eps \mathbf{Z}^1$ it follows that $\pi'\approx_\eps \pi$ and 
$(1-\pi') \approx_\eps (1-\pi)$ and therefore $p\in [0,1]$ for any applicable $0\leq \pi \leq \pi'\leq 1$.
\item
Let $\mathbf{C}$ be the distribution with point mass on $\bot=\{0\}$. 
\item
Let $\mathbf{B}^0 = \Ber(1-\frac{\pi'-\pi}{\pi'-e^{-\eps}\pi}) = \Ber(\frac{\pi-\pi e^{-\eps}}{\pi'-e^{-\eps}\pi})$
\item
Let $\mathbf{B}^1 = \Ber(1-\frac{\pi'-\pi}{e^\eps \pi'- \pi}) = \Ber(\frac{e^\eps \pi'- \pi'}{e^\eps \pi'- \pi})$
\end{itemize}

We show that this choice satisfies Definition~\ref{def:qtarget} with $q=\frac{1}{e^{\eps}+1}$.
\begin{itemize}
\item
We show that for both $b\in\{0,1\}$.
$\mathbf{Z}^b \equiv p\cdot \mathbf{C} + (1-p)\cdot \mathbf{B}^b$:
It suffices to show that the probability of $\bot$ is the same for the distributions on both sides.
For $b=0$, the probability of $\bot$ in the right hand side distribution is
\begin{align*}
p + (1-p)\cdot \frac{\pi'-\pi}{\pi'-e^{-\eps}\pi} &= 1- \frac{\pi'e^\eps-\pi}{e^\eps-1} + \frac{\pi'e^\eps-\pi}{e^\eps-1} \cdot \frac{\pi'-\pi}{\pi'-e^{-\eps}\pi}= 1-\pi\ .
\end{align*}
For $b=1$, the probability is
\begin{align*}
p + (1-p)\cdot\frac{\pi'-\pi}{e^\eps \pi'- \pi} &= 
1- \frac{\pi'e^\eps-\pi}{e^\eps-1} + \frac{\pi'e^\eps-\pi}{e^\eps-1} \cdot \frac{\pi'-\pi}{e^\eps \pi'- \pi}\\
&= 1- \frac{\pi'e^\eps-\pi}{e^\eps-1}\left(1- \frac{\pi'-\pi}{e^\eps \pi'- \pi}\right) \\
&= 1- \frac{\pi'e^\eps-\pi}{e^\eps-1}\cdot \frac{e^\eps \pi'- \pi-\pi'+\pi}{e^\eps \pi'- \pi}=1-\pi'\ .
\end{align*}
\item
We show that for $b\in\{0,1\}$, $\Pr[\mathbf{B}^b \in \top] \geq \frac{1}{e^{\eps}+1}$.
\begin{align*}
\Pr[\mathbf{B}^0 \in \top] &=  \frac{\pi - e^{-\eps}\pi}{\pi'-e^{-\eps}\pi} \\
&\geq \frac{\pi - e^{-\eps}\pi}{e^\eps \pi -e^{-\eps}\pi} =
\frac{e^{\eps}-1}{e^{2\eps}-1}=\frac{1}{e^{\eps}+1}\ .\\
\Pr[\mathbf{B}^1 \in \top] & = \frac{\pi'(e^\eps-1)}{\pi'e^\eps - \pi} \\ &\geq \frac{\pi(e^\eps-1)}{\pi e^{2\eps} - \pi} = \frac{e^{\eps}-1}{e^{2\eps}-1}=\frac{1}{e^{\eps}+1}
\end{align*}
Note that the inequalities are tight when $\pi'=\pi$ (and are tighter when $\pi'$ is closer to $\pi$).  This means that our selected $q$ is the largest possible that satisfies the conditions for the target being $\top$.
\item
We show that  $\mathbf{B}^0$ and $\mathbf{B}^1$ are $\eps$-indistinguishable, that is
\[\Ber(1-\frac{\pi'-\pi}{\pi'-e^{-\eps}\pi}) \approx_\eps
\Ber(1-\frac{\pi'-\pi}{e^\eps \pi'- \pi}) .
\]
Recall that $\Ber(a) \approx_\eps \Ber(b)$ if and only if
$a \approx_\eps b$ and $(1-a) \approx_\eps (1-b)$.
First note that
\[
e^{-\eps} \cdot \frac{\pi'-\pi}{\pi'-e^{-\eps}\pi} = \frac{\pi'-\pi}{e^\eps \pi'- \pi}
\]
Hence
\[
\frac{\pi'-\pi}{\pi'-e^{-\eps}\pi} \approx_\eps \frac{\pi'-\pi}{e^\eps \pi'- \pi}\ .
\]
It also holds that
\begin{align*}
1\leq \frac{\frac{\pi - e^{-\eps}\pi}{\pi'-e^{-\eps}\pi}}{\frac{\pi'(1- e^{-\eps})}{\pi'- e^{-\eps}\pi}} = \frac{\pi'}{\pi}\leq e^\eps .
\end{align*}
\end{itemize}

\end{proof}





\begin{proof} [Proof of Lemma~\ref{lemma:limitednotprior}]
The proof is very similar to that of Lemma~\ref{lemma:binaryFnP}, with a few additional details since $\top=\calY$ can have more than one element (recall that $\bot$ is a single element).

Assume (otherwise we switch roles) that $\Pr[\mathbf{Z^0}=\bot] \geq \Pr[\mathbf{Z}^1=\bot]$. Let
\begin{align*}
\pi &= \Pr[\mathbf{Z^0} \in\calY]\\
\pi' &= \Pr[\mathbf{Z^1}\in\calY]\ .
\end{align*}
Note that $\pi' \geq \pi$.

We choose $p$, $\mathbf{C}$, $\mathbf{B}^0$, $\mathbf{B}^1$ as follows.
Note that when $\pi \geq \frac{1}{e^{\eps}+1}$, then the choice of $p=0$ and $\mathbf{B}^b = \mathbf{Z}^b$ satisfies the conditions. 
Generally,
\begin{itemize}
\item
Let
\[
p=1- \frac{\pi'e^\eps-\pi}{e^\eps-1}\ .
\]
\item
Let $\mathbf{C}$ be the distribution with point mass on $\bot$.  
\item
Let $\mathbf{B}^0$ be $\bot$ with probability $\frac{\pi'-\pi}{\pi'-e^{-\eps}\pi}$ and otherwise
(with probability $\frac{\pi-\pi e^{-\eps}}{\pi'-e^{-\eps}\pi}$)
be  $\mathbf{Z}^0$ conditioned on the outcome  being in $\calY$.
\item
Let $\mathbf{B}^1$ be $\bot$ with probability $\frac{\pi'-\pi}{e^\eps \pi'- \pi}$ and otherwise 
(with probability
$\frac{e^\eps \pi'- \pi'}{e^\eps \pi'- \pi}$)
be $\mathbf{Z}^1$ conditioned on the outcome being in $\calY$.
\end{itemize}

It remains to show that these choices satisfy  Definition~\ref{def:qtarget}:

The argument for $\Pr[\mathbf{B}^b \in \calY] \geq \frac{e^{\eps}-1}{e^{2\eps}-1}$ is identical to Lemma~\ref{lemma:binaryFnP} (with $\calY=\top$).

We next verify that for $b\in\{0,1\}$:
$\mathbf{Z}^b \equiv p\cdot \mathbf{C} + (1-p)\cdot \mathbf{B}^b$.
The argument for the probability of $\bot$ is identical to 
Lemma~\ref{lemma:binaryFnP}. The argument for $y\in \calY$ follows from the probability of being in $\calY$ being the same and that proportions are maintained.

For $b=0$, the probability of $y\in \calY$ in the right hand side distribution is
\begin{align*}
(1-p)\cdot \frac{\pi-\pi e^{-\eps}}{\pi'-e^{-\eps}\pi} \cdot \frac{\Pr[\mathbf{Z}^0 = y]}{ \Pr[\mathbf{Z}^0\in\calY]} 
 &= \pi \cdot \frac{\Pr[\mathbf{Z}^0 = y]}{\Pr[\mathbf{Z}^0\in\calY]} = \Pr[\mathbf{Z}^0 = y] .
\end{align*}
For $b=1$, 
the probability of $y\in \calY$  in the right hand side distribution is
\begin{align*}
(1-p)\cdot \frac{e^\eps \pi'- \pi'}{e^\eps \pi'- \pi} \cdot \frac{\Pr[\mathbf{Z}^1 = y]}{\Pr[\mathbf{Z}^1\in\calY]} &= \pi' \cdot \frac{\Pr[\mathbf{Z}^1 = y]}{\Pr[\mathbf{Z}^1\in\calY]}\\
&= \Pr[\mathbf{Z}^1 = y] .
\end{align*}

Finally, we verify that $\mathbf{B}^0$ and $\mathbf{B}^1$ are $\eps$-indistinguishable.  Let $W\subset \calY$. We have
\begin{align*}
\Pr[\mathbf{B}^0\in W] &= \frac{\pi(1-e^{-\eps})}{\pi'-e^{-\eps}\pi}\cdot \frac{\Pr[\mathbf{Z}^0\in W] }{\pi} = \frac{e^\eps-1}{\pi'e^\eps-\pi} \Pr[\mathbf{Z}^0\in W] \\
\Pr[\mathbf{B}^1\in W] &= \frac{\pi'(e^\eps-1)}{e^\eps \pi'-\pi} \cdot \frac{\Pr[\mathbf{Z}^1\in W] }{\pi'} = \frac{e^\eps-1}{\pi'e^\eps-\pi} \Pr[\mathbf{Z}^1\in W]\ .
\end{align*}
Therefore
\begin{align*}
\frac{\Pr[\mathbf{B}^0\in W]}{\Pr[\mathbf{B}^1\in W]} = \frac{\Pr[\mathbf{Z}^0\in W]}{\Pr[\mathbf{Z}^1 \in W]}
\end{align*}
and we use $\mathbf{Z}^0 \approx_\eps \mathbf{Z}^1$.
The case of $W=\bot$ is identical to the proof of Lemma~\ref{lemma:binaryFnP}.  The case $\bot\in W$ follows.
\end{proof}

%
%

\section{Conditional Release with Revisions} \label{CondRelease:sec}

In this section we analyze an extension to 
conditional release that allows for revision calls to be made with respect to \emph{previous} computations. 
This extension was presented in Section~\ref{condreleaseintro:sec} and described in Algorithm~\ref{algo:conditionalrelease}.
A conditional release applies a private algorithm $\calA\to \calY$ with respect to a subset of outcomes $\top\subset \calY$. It draws $y\sim \calA(D)$ and returns $y$ if $y\in\top$ and $\bot$ otherwise.
Each revise calls effectively expands the target to $\top_h \cup \top'$, when $\top_h$ is the prior target and $\top'$ a disjoint extension.  If the (previously) computed result hits the expanded target ($y\in\top'$), the value $y$ is reported and charged. Otherwise, additional revise calls can be performed. The revise calls can be interleaved with other TCT computations at any point in the interaction.

\subsection{Preliminaries}

For a distribution $\mathbf{Z}$ with support $\calY$ and $W\subset \calY$ 
we denote by $\mathbf{Z}_W$ the distribution with support $W\cup\{\bot\}$ where outcomes not in $W$ are ``replaced'' by $\bot$.  That is, for $y\in W$,
$\Pr[\mathbf{Z}_W=y] := \Pr[\mathbf{Z} = y]$ and $\Pr[\mathbf{Z}_W=\bot] := \Pr[\mathbf{Z} \not\in W]$.

For a distribution $\mathbf{Z}$ with support $\calY$ and $W\subset \calY$ we denote by $\mathbf{Z} \mid W$
the \emph{conditional distribution} of $\mathbf{Z}$ on $W$.  That is, for $y\in W$, $\Pr[(\mathbf{Z} \mid W)=y] := \Pr[\mathbf{Z}=y]/\Pr[\mathbf{Z} \in W]$.
\begin{lemma} \label{lemma:basiccond}
  If $\mathbf{B}^0 \approx_{\eps,\delta} \mathbf{B}^1$ then
$\mathbf{B}^0_W \approx_{\eps,\delta} \mathbf{B}^1_W$.
\end{lemma}

\begin{lemma} \label{lemma:conddouble}
Let $\mathbf{B}^0$, $\mathbf{B}^1$ be probability distributions with support $\calY$ such that
$\mathbf{B}^0 \approx_\eps \mathbf{B}^1$.
Let $W\subset \calY$. Then
$\mathbf{B}^0\mid W  \approx_{2\eps} \mathbf{B}^1 \mid W$.
\end{lemma}

We extend these definitions to a randomized algorithm $\calA$, where 
$\calA_W(D)$ has distribution $\calA(D)_W$ and $(\calA\mid W)(D)$ has distribution $\calA(D)\mid W$.  The claims in Lemma~\ref{lemma:basiccond} and Lemma~\ref{lemma:conddouble}  then transfer to privacy of the algorithms.

\subsection{Analysis}

To establish correctness, it remains to show that each \CR call with an $(\eps,\delta)$-DP algorithm $\calA$ 
can be casted in TCT as a call to an $(\eps,\delta)$-DP algorithm with a \NotPrior target and each \RCR call cap be casted as a call to an $2\eps$-DP algorithm with a \NotPrior target. 
\begin{proof} [Proof of Lemma~\ref{lemma:analyze-revise}]
The claim for \CR was established in Lemma~\ref{CRprivacy:lemma}: 
Conditional release \CR$(\calA,\top)$ calls the algorithm $\calA_\top$ with target $\top$. From Lemma~\ref{lemma:basiccond}, $\calA_\top$ is $(\eps,\delta)$-DP when $\calA$ is $(\eps,\delta)$-DP.  $\top$ constitutes a \NotPrior target for $\calA_\top$ with respect to prior $\bot$.

We next consider revision calls as described in
Algorithm~\ref{algo:conditionalrelease}.
 We first consider the case of a pure-DP $\calA$ ($\delta=0$).
 
 When \CR publishes $\bot$, the internally stored value $r_h$ conditioned on published $\bot$ is a sample from the conditional distribution $\calA(D) \mid \neg\top$.

We will show by induction that this remains true after \RCR calls, that is the distribution of $r_h$ conditioned on $\bot$ being returned in all previous calls is $\calA(D) \mid \neg{\top_h}$ where $\top_h$ is the current expanded target.

An $\RCR$ call with respect to current target $\top_h$ and extension $\top'$ can be equivalently framed as drawing 
$r \sim \calA(D) \mid \neg\top_h$. 
From Lemma~\ref{lemma:conddouble}, if $\calA$ is $\eps$-DP then $\calA \mid \neg{\top_h}$ is $2\eps$-DP. 
If $r\in \top'$ we publish it and otherwise we publish $\bot$. This is a conditional release computation with respect to the $2\eps$-DP algorithm $\calA \mid \neg\top_h$ and the target $\top'$.  Equivalently, it is a call to the $2\eps$-DP algorithm $(\calA \mid \neg\top_h)_{\top'}$ with a \NotPrior target $\top'$.

Following the \RCR call, the conditional distribution of $r_h$ conditioned on $\bot$ returned in the previous calls is $\calA(D) \mid \neg (\top_h\cup \top')$ as claimed. We then update $\top_h \gets \top_h \cup \top'$.


It remains to handle the case $\delta>0$.
We consider \RCR calls for the case where $\calA$ is $(\eps,\delta)$-DP (approximate DP).  In this case, we want to show that we charge for the $\delta$ value once, only on the original \CR call.  We apply the simulation-based analysis in the proof of Theorem~\ref{thm:TCprivacy} with two fixed neighboring datasets.  Note that this can be viewed as each call being with a pair of distributions with an appropriate $q$-target (that in our case is always a \NotPrior target).

The first \CR call uses the distributions
$\calA(D^0)$ and $\calA(D^1)$.  From Lemma~\ref{relateapprox:lemma} they can be expressed as respective mixtures of pure 
$\mathbf{N}^0 \approx_\eps \mathbf{N}^1$ part (with probability $1-\delta$)
and non-private parts.  The non-private draw is designated failure with probability $\delta$.  Effectively, the call in the simulation is then applied to the pair $(\mathbf{N}^0_\top,\mathbf{N}^1_\top)$ with target $\top$.

A followup \RCR call is with respect to the previous target $\top_h$ and target extension $\top'$. 
The call is with the distributions $(\mathbf{N}^b \mid \neg\top_h)_{\top'}$ that using Lemma~\ref{lemma:basiccond} and Lemma~\ref{lemma:conddouble} satisfy
$(\mathbf{N}^0 \mid \neg\top_h)_{\top'} \approx_{2\eps} (\mathbf{N}^1 \mid \neg\top_h)_{\top'}$.
\end{proof}

\section{Boundary Wrapper Analysis} \label{sec:boundarywrapper}




In this section we provide details for the boundary wrapper method including proofs of 
Lemma~\ref{wrapperprivacy:lemma} and
Lemma~\ref{boundaryq:lemma}. For instructive reasons, we first
consider the special case of private testing and then outline the extensions to private classification.


Algorithm~\ref{algo:bwrapper} when specialized for tests first computes $\pi(D)=\min\{\Pr[\calA(D) = 0],1-\Pr[\calA(D) = 0]\}$, returns $\top$ with probability $\pi/(1+\pi)$ and otherwise (with probability $1/(1+\pi)$) return $\calA(D)$.  Overall, we return the less likely outcome with probability $\pi/(1+\pi)$, and the more likely one with probability
$(1-\pi)/(1+\pi)$.  

\begin{lemma} [Privacy of wrapped test] \label{testwrapperprivacy:lemma}
If the test is $\eps$-DP then the wrapper test is $t(\eps)$-DP where
$t(\eps)\leq \frac{4}{3}\eps$.
\end{lemma}
\begin{proof}
Working directly with the definitions, 
$t(\eps)$ is the maximum of 
\begin{align}
\max_{\pi\in(0,1/2)} \left|\ln\left(\frac{1-e^{-\eps}\pi}{1+e^{-\eps}\pi}\cdot \frac{1+\pi}{1-\pi} \right)\right| &\leq \frac{4}{3}\eps \label{bigstable:eq}\\
\max_{\pi\in(0,1/2)} \left|\ln\left(\frac{e^{-\eps}\pi}{1+e^{-\eps}\pi}\cdot \frac{1+\pi}{\pi} \right)\right| &\leq \eps \label{smallstable:eq} \\
\max_{\pi\in(\frac{e^{-\eps}}{2}, \frac{1}{1+e^\eps})} \left|\ln\left(\frac{\pi}{1+\pi}\cdot \frac{2-e^\eps\pi}{e^\eps \pi} \right)\right| &\leq \eps \label{smalllarge:eq} \\
\max_{\pi\in(\frac{e^{-\eps}}{2}, \frac{1}{1+e^\eps})} \left|\ln\left(\frac{1-\pi}{1+\pi}\cdot \frac{2-e^\eps\pi}{1-e^\eps \pi} \right)\right| &\leq \frac{4}{3}\eps \label{largesmall:eq}\\
\max_{\pi\in(\frac{e^{-\eps}}{2}, \frac{1}{1+e^\eps})} \left|\ln\left(\frac{\pi}{1+\pi}\cdot \frac{2-e^\eps\pi}{1-e^\eps \pi} \right)\right| &\leq \eps \label{boundaryflip:eq}
\end{align}
Inequality~\eqref{bigstable:eq}  bounds the
ratio change in the probably of the larger probability outcome when it remains the same and \eqref{smallstable:eq} the ratio change in the probability of the smaller probability outcome when it remains the same between the neighboring datasets.
When the less probable outcome changes between the neighboring datasets
it suffices to consider the case where the probability of the initially less likely outcome changes to
$e^\eps \pi > 1/2$ so that $e^\eps\pi < 1-\pi$, that is the change is from $\pi$ to $e^\eps\pi$ where
$\pi\in (\frac{e^{-\eps}}{2}, \frac{1}{1+e^\eps})$. Inequalities~\ref{smalllarge:eq} and~\ref{largesmall:eq}  correspond to this case. The wrapped probabilities of the $\top$ outcome are the same as the less probably outcome in the case that it is the same in the two databases.  Inequality~\ref{boundaryflip:eq}  corresponds to the case when there is change.
\end{proof}

We now show that $\top$ is a target for the wrapped test.
\begin{lemma} [$q$-value of the boundary target] \label{testboundaryq:lemma}
The outcome $\top$ of a boundary wrapper of an $\eps$-DP test is 
a $\frac{e^{t(\eps)} - 1}{2(e^{\eps+t(\eps)} - 1)}$-target.
\end{lemma}
\begin{proof}
Consider two neighboring datasets where the same outcome is less likely for both and 
$\pi \leq \pi'$.  Suppose without loss of generality that $0$ is the less likely outcome.

The common distribution $\mathbf(C)$ has point mass on $1$.

The distribution $\mathbf{B}^0$ is a scaled part of $\calM(D^0)$ that includes all $0$ and $\top$ outcomes (probability $\pi/(1+\pi)$ each) and 
probability of $\Delta \frac{e^{t(\eps)}}{e^{t(\eps)}-1}$ of the $1$ outcomes, where
$\Delta = \frac{2\pi'}{1+\pi'} - \frac{2\pi}{1+\pi}$.

The distribution $\mathbf{B}^1$ is a scaled part of $\calM(D^1)$ that includes all $0$ and $\top$ outcomes (probability $\pi'/(1+\pi')$ each) and 
probability of $\Delta \frac{1}{e^{t(\eps)}-1}$ of the $1$ outcomes.

It is easy to verify that $\mathbf{B}^0 \approx_{t(\eps)} \mathbf{B}^1$ and that
\begin{align*}
1-p &= \frac{2\pi'}{1+\pi'} + \Delta \frac{1}{e^{t(\eps)}-1} = \frac{2\pi}{1+\pi} + \Delta \frac{e^{t(\eps)}}{e^{t(\eps)}-1} \\
&= \frac{2\pi'}{1+\pi'}\frac{e^{t(\eps)}}{e^{t(\eps)}-1}  - \frac{2\pi}{1+\pi} \frac{1}{e^{t(\eps)}-1} \\ &= \frac{2}{e^{t(\eps)}-1} (e^{t(\eps)} \frac{\pi'}{1+\pi'} - \frac{\pi}{1+\pi} )  \\
\end{align*}

Using $\frac{\pi}{1+\pi} \leq \frac{\pi'}{1+\pi'}$  and
$\frac{\frac{\pi'}{1+\pi'}}{\frac{\pi}{1+\pi}}\leq e^{\eps} $
we obtain 
\begin{align*}
q &\geq \frac{\frac{\pi}{1+\pi}}{1-p} \\
&= \frac{e^{t(\eps)}-1}{2}\left(\frac{1}{e^{t(\eps)} \cdot \frac{\pi'}{1+\pi'} \cdot \frac{1+\pi}{\pi}  -1 } \right)\\
&\geq \frac{e^{t(\eps)}-1}{2} \frac{1}{e^{t(\eps)+\eps} -1} .
\end{align*}
\end{proof}

\paragraph{Extension to Private Classification}
To extension from Lemma~\ref{testwrapperprivacy:lemma} to Lemma~\ref{wrapperprivacy:lemma} follows by noting that the same arguments also hold respectively for sets of outcomes and also cover the case when there is no dominant outcome and when there is a transition between neighboring datasets from no dominant outcome to a dominant outcome.  The extension from Lemma~\ref{testwrapperprivacy:lemma} to Lemma~\ref{wrapperprivacy:lemma} is also straightforward by also noting the cases above (that only make the respective $\Delta$ smaller), and allowing $\mathbf{C}$ to be empty when there is no dominant outcome.

\section{Boundary wrapping without a probability oracle} \label{nooraclewrap:sec}
We present a boundary-wrapping method that does not assume a probability oracle.  This method accesses the distribution $\calA(D)$ in a blackbox fashion.

\begin{lemma}\label{lemma:run-twice}
Suppose $\calA:\calX^*\to\calY$ is an $(\eps,0)$-DP algorithm where $|\calY|<\infty$. Denote by $\calA\circ \calA$ the following algorithm: on input $D$, independently run $\calA$ twice and publish both outcomes. Define $E := \{(y,y'):y\ne y'\} \subseteq \calY\times \calY$. Then, $\calA\circ\calA$ is a $(2\eps,0)$-DP algorithm, and $E$ is a $f(\eps)$-target for $\calA\circ\calA$, where
\[
f(\eps) = 1-\sqrt{e^{2\eps}/(1+e^{2\eps})}.
\]
\end{lemma}

The algorithm behind Lemma~\ref{lemma:run-twice} is a natural one: we run the given mechanism $\calA$ twice, and pay the privacy cost only when the two outcomes are different. Intuitively, if there is a dominant outcome $i^*$ such that $\Pr[\calA(D) = i^*]\approx 1$, then both executions would output $i^*$ with high probability, in which case there will be no privacy cost.

\begin{proof}
$\calA\circ \calA$ is $(2\eps,0)$-DP by the basic composition theorem. Next, we verify the second claim. 

Identify elements of $\calY$ as $1,2,\dots, m=|\calY|$. Let $D,D'$ be two adjacent data sets. For each $i\in [m]$, let 
\[
p_i = \Pr[\calA(D) = i], ~~~ p'_i = \Pr[\calA(D') = i].
\]

We define a distribution $\mathbf{C}$. For each $i\in [m]$, define $q_i$ to be the largest real such that
\[
p_i^2 - q_i \in [e^{-2\eps} ({p'_i}^2 - q_i), e^{2\eps} ({p'_i}^2 - q_i)].
\]
Then, we define $\mathbf{C}$ to be a distribution over $\{(i,i):i\in [m]\}$ where $\Pr[\mathbf{C}=(i,i)]=\frac{q_i}{\sum_j q_j}$.

We can then write $(\calA\circ \calA)(D) = \alpha \cdot \mathbf{C} + (1-\alpha) \cdot \mathbf{N}^0$ and $(\calA\circ \calA)(D') = \alpha\cdot \mathbf{C} + (1-\alpha)\cdot \mathbf{N}^1$, where $\alpha = \sum_{i} q_i$, and $\mathbf{N}^0$ and $\mathbf{N}^1$ are $2\eps$-indistinguishable.

Next, we consider lower-bounding $\Pr[\mathbf{N}^0 = (y,y'):y\ne y']$. The lower bound of $\Pr[\mathbf{N}^0 = (y,y'):y\ne y']$ will follow from the same argument.

Indeed, we have
\[
\frac{\Pr[\mathbf{N}^0=(y,y'):y\ne y']}{\Pr[\mathbf{N}^0=(y,y)]} =  \frac{\sum_{i} p_i(1-p_i)}{\sum_{i} p_i^2 - q_i}.
\]
We claim that
\[
p_i^2 - q_i \le  1-{p'_i}^2.
\]
The inequality is trivially true if $p_i^2 \le 1-{p'_i}^2$. Otherwise, we can observe that for $q := {p_i}^2+{p'_i}^2 - 1 > 0$, we have ${p_i}^2-q = 1-{p'_i}^2$ and ${p'_i}^2 - q = 1-{p_i}^2$. Since $1-{p'_i}^2 \in [e^{-2\eps} (1-p_i^2), e^{2\eps} (1-p_i^2)]$, this implies that $q_i$ can only be larger than $q$.

Since we also trivially have that $p_i^2 - q_i \le {p_i}^2$, we conclude that
\[
\frac{\Pr[\mathbf{N}^0=(y,y'):y\ne y']}{\Pr[\mathbf{N}^0=(y,y)]} \ge  \frac{\sum_{i} p_i(1-p_i)}{\sum_{i} \min(p_i^2, 1-{p'_i}^2)} \ge \frac{\sum_{i} p_i(1-p_i)}{\sum_{i} \min(p_i^2, e^{2\eps}(1-p_i^2))}.
\]
Next, it is straightforward to show that, for every $p\in [0,1]$, one has
\[
\frac{ p(1-p)}{\min(p^2, e^{2\eps}(1-{p}^2))} = \min\left( \frac{1-p}{p}, \frac{p}{e^{2\eps}(1+p)} \right) \ge \frac{1-\sqrt{e^{2\eps}/(1+e^{2\eps})}}{\sqrt{e^{2\eps}/(1+e^{2\eps})}}.
\]
Consequently,
\[
\Pr[\mathbf{N}^0=(y,y'):y\ne y'] =  \frac{\Pr[\mathbf{N}^0=(y,y'):y\ne y']}{\Pr[\mathbf{N}^0=(y,y'):y\ne y']+\Pr[\mathbf{N}^0=(y,y)]} \ge 1-\sqrt{e^{2\eps}/(1+e^{2\eps})},
\]
as desired.
\end{proof}

\begin{remark}
For a typical use case where $\epsilon = 0.1$, we have $f(\eps)\approx 0.258$. Then, by applying Theorem~\ref{thm:TCprivacy}, on average we pay $\approx 8\eps$ privacy cost for each target hit. Improving the constant of $8$ is a natural question for future research. 
We also note that while the overhead is more significant compared to the
boundary wrapper of Algorithm~\ref{algo:bwrapper}, the output is more informative as it includes two independent responses of the core algorithm whereas Algorithm~\ref{algo:bwrapper} returns one or none (when $\top$ is returned).
We expect that it is possible to design less-informative boundary wrappers for the case of blackbox access (no probability oracle) that have a lower overhead. We leave this as an interestion question for followup work.
\end{remark}





\section{$q$ value for \texttt{BetweenThresholds}} \label{qvotbetween:sec}

We provide details for the \BetweenThresholds classifier (see Section~\ref{betweenintro:sec}).
The \BetweenThresholds classifier 
is a refinement of \AboveThreshold.
It is specified by a
$1$-Lipschitz function $f$, two thresholds $t_\ell < t_r$, and a privacy parameter $\eps$.
We compute $\tilde{f}(D) = f(D) + \Lap(1/\eps)$, where $\Lap$ is the Laplace distribution.  If 
$\tilde{f}(D) < t_\ell$ we return \texttt{L}. If
$\tilde{f}(D) > t_r$ we return \texttt{H}. Otherwise, we return $\top$.

\begin{lemma}[Effectiveness of the ``between'' target] \label{qbetween:lemma}
The $\top$ outcome is an 
$(1-e^{-(t_r-t_l)\eps})\cdot \frac{e^{\eps}-1}{e^{2\eps}-1}$-target for 
\BetweenThresholds.
\end{lemma}
\begin{proof}
Without loss of generality we assume that 
$t_\ell=0$ and $t_r=t/\eps$.

Consider two neighboring data sets $D^0$ and $D^1$ and the respective $f(D^0)$ and $f(D^1)$. Since $f$ is $1$-Lipschitz, we can assume without loss of generality (otherwise we switch the roles of the two data sets) that
$f(D^0)\leq f(D^1) \leq f(D^0)+1$.
Consider the case $f(D^1) \leq 0$. The case 
$f(D^0) \geq t/\eps$ is symmetric and the cases where one or both of $f(D^b)$ are in $(0,t/\eps)$ make $\bot$ a more effective target.
\begin{align*}
\pi_L^b &:= \Pr[f(D^b)+ \Lap(1/\eps) < t_\ell=0] = 1-\frac{1}{2}e^{-|f(D^b)|\eps} \\
\pi_H^b &:= \Pr[f(D^b)+ \Lap(1/\eps) > t_r=t/\eps] = \frac{1}{2}e^{-(|f(D^b)|\eps -t} \\
\pi_\top^b &:= \Pr[f(D^b)+ \Lap(1/\eps) \in (0,t/\eps)] = \frac{1}{2} \left( e^{-|f(D^b)|\eps}- e^{-(|f(D^b)|\eps -t} \right)= \frac{1}{2} e^{-|f(D^b)|\eps}(1-e^{-t})
\end{align*}
Note that $\pi_L^0 \approx_\eps \pi_L^1$ and 
$\pi_H^1 \approx_\eps \pi_H^0$,
$\pi_L^0 \geq \pi_L^1$ and 
$\pi_H^1 \geq \pi_H^0$

We set 
\[
p = (\pi_L^1- \frac{1}{e^{\eps}-1}(\pi_L^0-\pi_L^1)) + (\pi_H^0-\frac{1}{e^{\eps}-1}(\pi_H^1-\pi_H^0))\]
and the distribution
$\mathbf{C}$ to be \texttt{L} with probability $(\pi_L^1- \frac{1}{e^{\eps}-1}(\pi_L^0-\pi_L^1))/p$ and 
\texttt{H} otherwise. 

We specify $p$ and the distributions $\mathbf{B}^b$ and
$\mathbf{C}$ as we did for \NotPrior (Lemma~\ref{lemma:NotPriorprivacy}) with respect to ``prior'' \texttt{L}. (We can do that and cover also the case where $f(D^0)>t/\eps$ where the symmetric prior would be \texttt{H} because the target does not depend on the values being below or above the threshold).

The only difference is that our target is smaller, and includes only $\top$ rather than $\top$ and \texttt{H}.
Because of that, the calculated $q$ value is reduced by a factor of
\begin{align*}
\frac{\pi_\top^b}{\pi_\top^b+\pi_H^b} = \frac{\frac{1}{2} e^{-|f(D^b)|\eps}(1-e^{-t})}{\frac{1}{2} e^{-|f(D^b)|\eps}}=(1-e^{-t})\ .
\end{align*}

%
%
\end{proof}

\section{Analysis of SVT with individual privacy charging} \label{SVTindividual:sec}

We provide the privacy analysis for SVT with individual privacy charging (see Section~\ref{SVTindividualintro:sec}).

\begin{proof} [Proof of Theorem~\ref{SVTindividual:thm}]
We apply simulation-based privacy analysis (see Section~\ref{lemma:intro-simulate}).  Consider two neighboring datasets $D$ and $D'=D\cup\{x\}$.  The only queries where potentially $f(D)\not= f(D')$ and we may need to call the data holder are those with $f(x)\not=0$. Note that for every $x'\in D$, the counter $C_{x'}$ is the same during the execution of Algorithm~\ref{algo:svt-individual} on either $D$ or $D'$. This is because the update of $C_{x'}$ depends only on the published results and $f_i(x')$, both of which are public information. Hence, we can think of the processing of $C_{x'}$ as a post-processing when we analyze the privacy property between $D$ and $D'$. 

After $x$ is removed, the response on $D$ and $D'$ is the same, and the data holder does not need to be called. Before $x$ is removed from $D'$, we need to consider the queries such that $f(x)\ne 0$ while $C_x< \tau$. Note that this is equivalent to a sequence of \AboveThreshold tests to linear queries, we apply TCT analysis with \CR applied with above threshold responses. The claim follows from Theorem~\ref{thm:TCprivacy}.
\end{proof}
We also add that Algorithm~\ref{algo:svt-individual} can be implemented with \BetweenThresholds test (see Section~\ref{betweenintro:sec}), the extension is straightforward with the respective privacy bounds following from Lemma~\ref{qbetween:lemma} ($q$ value for target hit).

\section{Private Selection} \label{sec:selection}


\newcommand{\Best}{\mathrm{Best}}

In this section we provide proofs and additional details for private selection in TCT (Sections~\ref{topkselectintro:sec} and~\ref{appselectionintro:sec}).
Let $\calA_1,\dots, \calA_m$ be of $m$ private algorithms that return results with quality scores. The private selection task asks us to select the best algorithm from the $m$ candidates.  The one-shot selection described in Algorithm~\ref{algo:top-k} (with $k=1$) runs each algorithm once and returns the response with highest quality.

It is shown in \cite{LiuT19-private-select} that if each $\calA_i$ is
$(\eps,0)$-DP then the one-shot selection algorithm degrades the privacy bound to $(m\eps, 0)$-DP. However, if we relax the requirement to approximate DP, we can show that one-shot selection is $(O(\log (1/\delta)\eps),\delta)$-DP, which is independent of $m$ (the number of candidates).  
Moreover, in light of a lower-bound example by \cite{LiuT19-private-select}, Theorem~\ref{thm:vanilla-selection} is tight up to constant factors. 

Formally, our theorem can be stated as
\begin{theorem}\label{thm:vanilla-selection}
Suppose $\eps < 1$. Let $\calA_1,\dots, \calA_m:X^n \to \calY\times \mathbb{R}$ be a list of $(\eps, \delta_i)$-DP algorithms, where the output of $\calA_i$ consists of a solution $y\in \calY$ and a score $s\in \mathbb{R}$. Denote by $\Best(\calA_1,\dots, \calA_m)$ the following algorithm (Algorithm~\ref{algo:top-k} with $k=1$): run each $\calA_1,\dots, \calA_m$ once, get $m$ results $(y_1,s_1),\dots, (y_m,s_m)$, and output $(y_{i^*},s_{i^*})$ where $i^* = \arg\max_{i} s_i$. 

Then, for every $\delta \in (0,1)$, $\Best(\calA_1,\dots, \calA_m)$ satisfies $(\eps',\delta')$-DP where $\eps' = O(\eps\log(1/\delta)), \delta' = \delta + \sum_{i} \delta_i$.
\end{theorem}
\begin{proof}

\medskip\noindent\textbf{Discrete scores.} We start by considering the case that the output scores from $\calA_1,\dots, \calA_m$ always lie in a \emph{finite} set $X\subseteq \mathbb{R}$. The case with continuous scores can be analyzed by a discretization argument.

Fix $D^0, D^1$ to be a pair of adjacent data sets. We consider the following implementation of the vanilla private selection.

\begin{algorithm2e}[h]
    \caption{Private Selection: A Simulation}
    \label{algo:vanilla-selection}
    \DontPrintSemicolon
    \KwIn{
        Private data set $D$. The set $X$ defined above.
    }
    \For{$i= 1, \dots, m$}{
        $(y_i, s_i) \gets \calA_i(D)$ \;
    }
    \For{$\hat{s}\in X$ in the decreasing order}{
        \For{$i=1,\dots, m$}{
            \If{$s_i \ge \hat{s}$}{
                \KwRet{$(y_i,s_i)$} \;
            }
        }
    }
\end{algorithm2e}

Assuming the score of $\calA_i(D)$ always lies in the set $X$, it is easy to see that Algorithm~\ref{algo:vanilla-selection} simulates the top-1 one-shot selection algorithm (Algorithm~\ref{algo:top-k} with $k=1$) perfectly. Namely, Algorithm~\ref{algo:vanilla-selection} first runs each $\calA_i(D)$ once and collects $m$ results. Then, the algorithm searches for the \emph{lowest} $\hat{s}\in X$ such that there is a pair $(y_i,s_i)$ with a score of at least $s_i \ge \hat{s}$. The algorithm then publishes this score.

On the other hand, we note that Algorithm~\ref{algo:vanilla-selection} can be implemented by the conditional release with revisions framework (cf. Algorithm~\ref{algo:conditionalrelease}). Namely, Algorithm~\ref{algo:vanilla-selection} first runs each private algorithm once and stores all the outcomes. Then the algorithm gradually extends the target set (namely, when the algorithm is searching for the threshold $\hat{s}$, the target set is $\{(y,s):s\ge \hat{s}\}$), and tries to find an outcome in the target. Therefore, it follows from Lemma~\ref{lemma:analyze-revise} and Theorem~\ref{thm:TCprivacy} that Algorithm~\ref{algo:vanilla-selection} is $(O(\eps\log(1/\delta)),\delta+\sum_i \delta_i)$-DP.

\medskip\noindent\textbf{Continuous scores.} We then consider the case that the distributions of the scores of $\calA_1(D),\dots, \calA_K(D)$ are \emph{continuous} over $\mathbb{R}$. We additionally assume that the distribution has no ``point mass''. This is to say, for every $i\in [m]$ and $\hat{s}\in \mathbb{R}$, it holds that
\[
\lim_{\Delta \to 0}\Pr_{(y_i,s_i)\sim \calA_i(D)}[\hat{s}-\Delta \le s\le \hat{s}+\Delta] = 0.
\]
This assumption is without loss of generality because we can always add a tiny perturbation to the original output score of $\calA_i(D)$.

Fix $D,D'$ as two neighboring data sets. We show that the vanilla selection algorithm preserves differential privacy between $D$ and $D'$.

Let $\eta > 0$ be an arbitrarily small real number. Set $M = \frac{10\cdot m^4}{\eta}$. For each $\ell\in [1,M]$, let $q_{\ell} \in \mathbb{R}$ be the unique real such that
\[
\Pr_{i\sim [m], (y_i,s_i)\sim \calA_i(D)}[s_i \ge q_{\ell}] = \frac{\ell}{M+1}.
\]
Similarly we define $q'_{\ell}$ with respect to $\calA_i(D')$. Let $X=\{q_{\ell}, q'_{\ell}\}$.

Now, consider running Algorithm~\ref{algo:vanilla-selection} with the set $X$ and candidate algorithms $\calA_1,\dots, \calA_K$ on $D$ or $D'$. Sort elements of $X$ in the increasing order, which we denote as $X = \{\hat{q}_1\le \dots \le \hat{q}_m\}$. After sampling $\calA_i(D)$ for each $i\in [m]$, Algorithm~\ref{algo:vanilla-selection} fails to return the best outcome only if one of the following events happens.
\begin{itemize}
\item The best outcome $(y^*,s^*)$ satisfies that $s^* < \hat{q}_1$.
\item There are two outcomes $(y_i,s_i)$ and $(y_j,s_j)$ such that $s_i,s_j\in [\hat{q}_{\ell},\hat{q}_{\ell+1})$ for some $\ell\in [n]$.
\end{itemize}
If Item 1 happens, Algorithm~\ref{algo:vanilla-selection} does not output anything. If Item 2 happens, then it might be possible that $i < j, s_i > s_j$, but Algorithm~\ref{algo:vanilla-selection} outputs $s_i$.

It is easy to see that Event~1 happens with probability at most $\frac{m^2}{M}\le \eta$ by the construction of $X$. Event~2 happens with probability at most $M\cdot \frac{m^4}{M^2}\le \eta$. Therefore, the output distribution of Algorithm~\ref{algo:vanilla-selection} differs from the true best outcome by at most $O(\eta)$ in the statistical distance. Taking the limit $\eta \to 0$ completes the proof.
\end{proof}

\begin{remark}
Theorem~\ref{thm:vanilla-selection} shows that there is a factor of $\log(1/\delta)$ overhead when we run top-1 one-shot private selection (Algorithm~\ref{algo:vanilla-selection}) only once. Nevertheless, we observe that if we compose top-1 one-shot selection
with other algorithms under the TCT framework (e.g., compose multiple top-1 one-shot selections, generalized private testing, or any other applications mentioned in this paper)), then \emph{on-average} we only pay $4\eps$ privacy cost 
(one \NotPrior target hit with a $2\eps$-DP algorithm)
per top-1 selection (assuming $\eps$ is sufficiently small so that $e^{\eps} \approx 1$). In particular, adaptively performing $c$ executions of top-1 selection 
is $(\eps',\delta)$-DP where $\eps' = \eps \cdot (4\sqrt{c\log(1/\delta)}+o(\sqrt{c}))$.

Liu and Talwar \cite{LiuT19-private-select} established a lower bound of $2\eps$ on the privacy of a more relaxed top-1 selection task.  Hence, there is a factor of 2 gap between this lower bound and our privacy analysis.  Note that for the simpler task one-shot above threshold score (discussed in Section~\ref{appselectionintro:sec}), where the goal is to return a response that is above the threshold if there is one, can be implemented using a single target hit on Conditional Release call (without revise) and this matches the lower bound of $2\eps$. We therefore suspect that it might be possible to tighten the privacy analysis of top-1 one-shot selection.  We leave it as an interesting question for followup work. 
\end{remark}

\subsection{One-Shot Top-$k$ Selection}

In this section, we prove our results for top-$k$ selection.

We consider the natural one-shot algorithm for top-$k$ selection described in Algorithm~\ref{algo:top-k}, which (as mentioned in the introduction) generalizes the results presented in \cite{DBLP:conf/nips/DurfeeR19,DBLP:conf/icml/QiaoSZ21}, which were tailored for selecting from $1$-Lipschitz functions, using the Exponential Mechanism or the Report-Noise-Max paradigm. 

We prove the following privacy theorem for Algorithm~\ref{algo:top-k}.

\begin{theorem}\label{thm:top-k-privacy}
Suppose $\eps < 1$. Assume that each $\calA_i$ is $(\eps,0)$-DP. Then, for every $\delta \in (0,1)$, Algorithm~\ref{algo:top-k} is $(\eps \cdot O(\sqrt{k\log(\frac{1}{\delta})} +\log(\frac{1}{\delta})),\delta)$-DP.
\end{theorem}

\begin{remark}
The constant hidden in the big-Oh depends on $\eps$. For the setting that $\eps$ is close to zero so that $e^\eps\approx 1$ and $\delta \ge 2^{o(k)}$, the privacy bound is roughly $(\eps',\delta)$-DP where $\eps' = \eps\cdot (4\sqrt{k\log(1/\delta)} + o(\sqrt{k}))$.
\end{remark}

\begin{remark}
We can take $\calA_i$ as the Laplace mechanism applied to a $1$-Lipschisz quality function $f_i$ (namely, $\calA_i(D)$ outputs  a pair $(i, f_i(D) + \Lap(1/\eps))$, where $i$ denotes the ID of the $i$-th candidate, and $f_i(D) + \Lap(1/\eps)$ is the noisy quality score of Candidate $i$ with respect to the data $D$). In this way, Theoerem~\ref{thm:top-k-privacy} recovers the main result of \cite{DBLP:conf/icml/QiaoSZ21} easily.

Moreover, Theorem~\ref{thm:top-k-privacy} improves over \cite{DBLP:conf/icml/QiaoSZ21} from three aspects: Firstly, Theorem~\ref{thm:top-k-privacy} allows us to report the noisy quality scores of selected candidates for free, while \cite{DBLP:conf/icml/QiaoSZ21} needs to run one additional round of Laplace mechanism to publish the quality scores. Second, our privacy bound has no dependence on $m$, while the bound in the prior work \cite{DBLP:conf/icml/QiaoSZ21} was $(O(\eps \sqrt{k\log(m/\delta)}),\delta)$-DP. Lastly, Theorem~\ref{thm:top-k-privacy} applies more generally to \emph{any} private-preserving algorithms, instead of the classic Laplace mechanism.
\end{remark}


\begin{proof}
The proof is similar to that of Theorem~\ref{thm:vanilla-selection}. Namely, we run each $\calA_i(D)$ once and store all results. Then we maintain a threshold $T$, which starts with $T=\infty$. We gradually decrease $T$, and use Algorithm~\ref{algo:conditionalrelease} (Conditional Release with Revised Calls) to find outcomes with a quality score larger than $T$. We keep this process until we identify $k$ largest outcomes. The claimed privacy bound now follows from Lemma~\ref{lemma:analyze-revise} and Theorem~\ref{thm:TCprivacy}.
\end{proof}


\end{document}